\documentclass[journal]{IEEEtran}

\usepackage{hyperref}
\usepackage{epsf}
\usepackage{graphicx}
\usepackage{amsfonts,amssymb,amsmath}
\usepackage{algorithmic,algorithm}
\usepackage{subfigure}
\usepackage{latexsym}
\usepackage{bm}
\usepackage[usenames]{color}
\usepackage{url}
\usepackage[normalem]{ulem}

\newcommand{\revise}[1]{{#1}}
\newcommand{\expect}[1]{\mathbb{E}\left[#1\right]}

\newtheorem{theorem}{Theorem}
\newtheorem{lemma}{Lemma}
\newtheorem{definition}{Definition}
\newtheorem{assumption}{Assumption}
\newtheorem{corollary}{Corollary}
\newtheorem{remark}{Remark}
\newtheorem{Alg}{Algorithm}

\begin{document}
\title{Data Center Server Provision: Distributed Asynchronous Control for Coupled Renewal Systems}
\author{\authorblockN{Xiaohan Wei and Michael J. Neely, \textit{Senior Member, IEEE}}
\thanks{The authors are with the  Electrical Engineering department at the University of Southern California, Los Angeles, CA.}
}
\maketitle

\begin{abstract}
This paper considers a cost minimization problem for data centers with $N$ servers and randomly arriving service requests.  A central router decides which server to use for each new request.  Each server has three types of states (active, idle, setup) with different costs and time durations.  The servers operate asynchronously over their own states and can choose one of multiple sleep modes when idle. We develop an online distributed control algorithm so that each server makes its own decisions, the request queues are bounded and the overall time average cost is near optimal with probability 1.  The algorithm does not need probability information for the arrival rate or job sizes. Next, an improved algorithm that uses a single queue is developed via a ``virtualization'' technique which is shown to provide the same (near optimal) costs. Simulation experiments on a real data center traffic trace demonstrate the efficiency of our algorithm compared to other existing algorithms.
\end{abstract}

\section{Introduction}
\subsection{Overview}
Consider a data center that consists of a central controller and $N$ servers that serve randomly arriving requests. The system operates in slotted time with time
slots $t \in \{0, 1, 2, \ldots\}$.  Each server $n \in \{1, \ldots, N\}$ has three basic states:
\begin{itemize}
  \item Active: The server is available to serve requests. Server $n$ incurs a cost of
  $e_n \geq 0$ on every active slot, regardless of whether or not requests are available to serve. \revise{In data center applications, such cost often represents the power consumption of each individual server.} 

  \item Idle: A low cost sleep state where no requests can be served. The idle state is actually
  comprised of a choice of multiple sleep modes with different per-slot costs. The specific sleep mode also affects the setup time required to transition from the idle state to the active state. For the rest of the paper, we use ``idle'' and ``sleep'' exchangeably.
  \item Setup: A transition period from idle to active during which no requests can be served.
  The setup cost and duration depend on the preceding sleep mode.   The setup duration is typically more than one slot, and can be a random variable that depends on the server $n$ and on the
  preceding sleep mode.
 \end{itemize}

 An active server can choose to transition to the idle state at any time.  When it does
so, it chooses the specific sleep mode to use and the amount of time to sleep.
For example, deeper sleep modes can shut down more electronics and thereby save
on per-slot idling costs.  However, a deeper sleep incurs a longer setup time when
transitioning back to the active state.  Each server makes separate
decisions about when to transition and what sleep mode to use.  The resulting
transition times for each server are asynchronous.   On top of this, a central controller
makes slot-wise decisions for routing requests to servers.  It can also reject requests
(with a certain amount of cost) if it decides they cannot be
supported. The goal is to minimize the overall time average cost.

This  problem is challenging mainly for two reasons: First, since each setup state generates cost but serves no request, it is not clear whether or not transitioning to idle from the
active state indeed saves power.
It is also not clear which sleep mode the server should switch to. Second, if one server is currently in a setup state, it cannot make another decision until it reaches the active state (which typically takes more than one slot), whereas other active servers can make decisions during this time.  Thus, this problem  can be viewed as a system with coupled Markov decision processes (MDPs) making decisions asynchronously. Classical methods for MDPs,  such as dynamic programming and linear programming \cite{dynamic_programming}\cite{puterman}\cite{ross-prob},  can be
impractical for two reasons:
First, the state space has dimension that depends on the number of servers, making solutions difficult  when the number of servers is large.  Second, some statistics of the system, such as the arrival probabilities, are unknown.

\subsection{Related works}
In this paper, we use renewal-reward theory (e.g. \cite{renewal_reward}) together with Lyapunov optimization (e.g. \cite{Lyapunov_optimization} and \cite{Lyapunov_book}) to develop a simple implementable algorithm for this problem. Our work is not alone in approaching the problem this way.
Work in \cite{frame_algorithm} uses Lyapunov theory for a system with one renewal server, while
 work in  \cite{deterministic_algorithm} considers a multi-server system but in a deterministic context. The work in \cite{two_stage_algorithm} develops a two stage algorithm for stochastic
 multi-server systems, but the first stage must be solved offline. This paper is distinct in that we overcome the open challenges of  previous work and develop a near optimal fully online algorithm for a stochastic system with multiple servers.

Several alternative approaches to multi-server systems use queueing theory.
Specifically,
the work in \cite{mmk_queueing_system} treats the multi-server system as an $M/M/k/setup$ queue, whereas the work in \cite{c_mu_rule} considers modeling a
single server system as a multi-class $M/G/1$ queue.  These approaches require Poisson traffic and do not treat the asynchronous control problem with multiple sleep options considered in the current paper.

Experimental work on power and delay minimization is treated in  \cite{Gandhi_phd_thesis}, which
proposes to turn each server ON and OFF according to the rule of an $M/M/k/setup$ queue. The work in \cite{virtual_data_center} applies Lyapunov optimization to optimize power in
virtualized data centers.  However, it assumes each server has negligible setup time and that
ON/OFF decisions are made synchronously at each server.
The works  \cite{across_data_centers}, \cite{right_sizing}
focus on power-aware provisioning over a time scale large enough so that the whole data center can adjust its service capacity. Specifically,
\cite{across_data_centers} considers load balancing across geographically distributed data centers,
and \cite{right_sizing} considers provisioning over a finite time interval and introduces an online 3-approximation algorithm.

\revise{
Prior works \cite{sleep-multi-mode, sleep-multi-mode-2, sleep-multi-mode-3}
consider servers with multiple hypothetical sleep states with different level of power consumption and setup times. Although empirical evaluations in these works show significant power saving by introducing sleep states, they are restricted to the scenario where the setup time from sleep to active is on the order of milliseconds, which is not realistic for today's data center. 
Realistic sleep states with setup time on the order of seconds are considered in \cite{sleep-mode-2}, where effective heuristic algorithms are proposed and evaluated via extensive testbed simulations.
However, little is known about the theoretical performance bound regarding these algorithms. 
In this paper, we propose a new algorithm which incorporates servers with different sleep modes while having a provable performance guarantee.}

\subsection{Contributions}
On the theoretical side, the current paper is the first to consider the stochastic control
problem with heterogeneous servers and multiple idle and setup states at each server.
This gives rise to a nonstandard asynchronous problem with coupled Markov decision
systems.
A novel aspect of the solution is the construction of a process with super-martingale properties by piecing together the asynchronous processes at each server.  This is interesting because neither the individual server processes nor their asynchronous sum are super-martingales. The technique yields a simple distributed algorithm that can likely be applied more broadly for other
coupled stochastic MDPs.

\revise{
On the practical side, we run our algorithm on a real data center traffic trace with server setup time on the order of seconds. Simulation experiments show that our algorithm is effective compared to several existing algorithms. 
}
%The structure of the paper is as follows: Section \ref{section_problem_formulation} defines the multi-server system and formulates our cost minimization problem. Section \ref{section_iid_algorithm} introduces a benchmark algorithm from dynamic programming which is theoretically important but practically impossible to implement. Section \ref{section_proposed_algorithm} proposes our online asynchronous control algorithm. Section \ref{section_performance_analysis} proves both the feasibility and near optimality of the proposed algorithm. Section \ref{section_data_center} implements the proposed algorithm to data center server provision. Finally, section \ref{section_simulation} demonstrates the the performance of the algorithm through simulation experiments.

\section{Problem formulation}\label{section_problem_formulation}
Consider a slotted time system with  $N$ servers, denoted by set $\mathcal{N}$,  that
serve randomly incoming requests.

\subsection{Front-end load balancing}
At each time slot $t \in \{0, 1, 2, \ldots\}$, $\lambda(t)$ new requests arrive at the system. We assume $\lambda(t)$ takes values in a finite set $\Lambda$.
Let  $R_n(t),~n\in\mathcal{N}$ denote the number of requests routed into server $n$ at time $t$.
In addition, the system is allowed to reject requests. Let $d(t)$ be the number of requests that are
rejected on slot $t$, and let $c(t)$ be the corresponding per-request cost for such rejection.
Assume $c(t)$ takes values in a finite state space $\mathcal{C}$. The $R_n(t)$ and $d(t)$ decision variables on slot $t$ must be nonnegative integers that satisfy:
 \begin{align*}
 &\sum_{n=1}^NR_n(t)+d(t)=\lambda(t)\\
 &\sum_{n=1}^N R_n(t)\leq R_{\max}
  \end{align*}
 for a given  integer $R_{max}>0$.  The vector process $(\lambda(t), c(t))$ takes values in $\Lambda \times \mathcal{C}$ and is assumed to be an independent and identically distributed (i.i.d.) vector over slots $t \in \{0, 1, 2, \ldots\}$ with an unknown probability mass function.

 Each server $n$ maintains a request queue $Q_n(t)$ that stores the requests that are
 routed to it.  Requests are served in  a FIFO manner with queueing dynamics as follows:
\begin{equation}\label{queue_update}
Q_n(t+1)=\max\left\{Q_n(t)+R_n(t)-\mu_n(t)H_n(t),~0\right\}.
\end{equation}
where $H_n(t)$ is an indicator variable that is 1 if server $n$ is active on slot $t$, and $0$ else, and $\mu_n(t)$ is a random variable that represents the number of requests can be
served on slot $t$.  Each queue is initialized to $Q_n(0)=0$.
Assume that, every slot in which server $n$ is active, $\mu_n(t)$ is independent and identically distributed with a known mean $\mu_n$. This randomness can model variation in job sizes.

\begin{assumption}\label{observability}
The process $\{(\lambda(t), c(t))\}_{t=0}^{\infty}$ is observable, i.e. the router
can observe the $(\lambda(t),c(t))$  realization each time slot $t$ before making decisions. In contrast, the process $\{\mu_n(t)\}_{t=0}^{\infty}$ is not observable, i.e. given that $H_n(t)=1$, at the beginning of slot $t$ server $n$ knows a random service will take place, but it does not know the
realization of $\mu_n(t)$ until the end of slot $t$. \revise{Moreover, $\lambda(t),~c(t)$ and $\mu_n(t)$ are all bounded.}
\end{assumption}

\subsection{Server model}
Each server $n\in \mathcal{N}$ has three types of states: active, idle, and setup.
The idle state of each server $n$ is further decomposed into a collection of distinct sleep modes.  Each server $n \in \mathcal{N}$ makes decisions over its own \emph{renewal frames}. Define the renewal frame for server $n$ as the time period between successive visits to active state (with each renewal period ending in an active state). Let $T_n[f]$ denote the frame size of the $f$-th renewal frame for server $n$, for  $f \in \{0, 1, 2, \ldots\}$. Let $t^{(n)}_f$ denote the start of frame $f$, so that
 $T_n[f]=t^{(n)}_{f+1}-t^{(n)}_f$.  Assume that $t^{(n)}_0=0$ for all $n \in \mathcal{N}$, so that time slot $0$ the start of the first renewal frame (labeled frame $f=0$) for all servers.  For simplicity, assume all servers are ``active'' on slot $t=-1$. Thus, the slot just before each renewal frame is an active slot.
Fig. \ref{fig:single-system} illustrates this renewal frame construction.

\begin{figure}[htbp]
   \centering
   \includegraphics[height=1in]{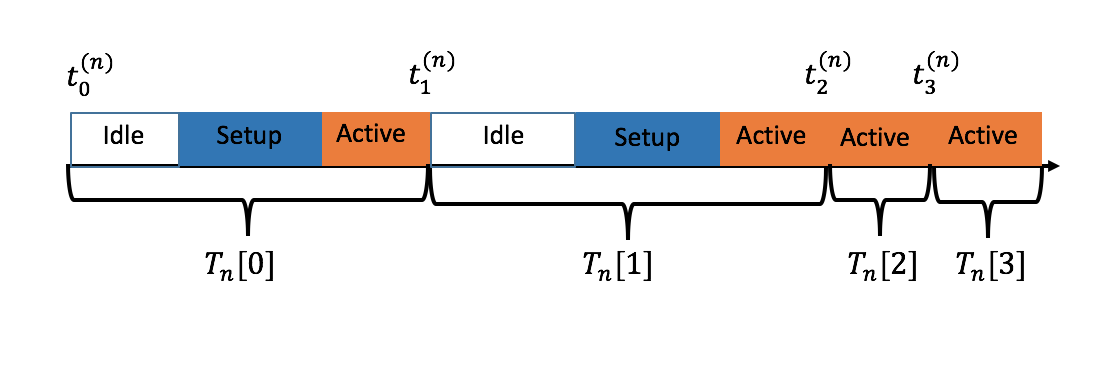} % requires the graphicx package
   \caption{Illustration of a typical renewal frame construction, where $T_n[i]$ is the length of frame $i$ and $t^{(n)}_i$ is the start slot of frame $i$.}
   \label{fig:single-system}
\end{figure}

 Fix a server $n \in \mathcal{N}$ and a frame index $f \in \{0, 1,2, \ldots\}$.  Time $t_f^{(n)}$ marks the start of renewal frame $f$. At this time, server $n$  must decide whether to remain active or to
go idle. If it remains active then the renewal frame lasts for one slot, so that $T_n[f]=0$.   If it goes idle, it chooses an idle mode from a finite set $\mathcal{L}_n$, representing the set of idle modes options. Let $\alpha_n[f]$ represent this initial decision for server $n$ at the start of frame $f$, so that:
\[ \alpha_n[f] \in \{active\}\cup\mathcal{L}_n \]
where $\alpha_n[f] =active$ means the server chooses to remain active.  If the server chooses to go idle, so that $\alpha_n[f] \in \mathcal{L}_n$, it then chooses a variable $I_n[f]$ that represents \emph{how much time it remains idle}. The
decision variable $I_n[f]$ is chosen as an integer
in the set $\{1, \ldots, I_{max}\}$ for some given integer $I_{max}>0$.  The consequences of these decisions are described below.

\begin{itemize}
\item Case $\alpha_n[f]=active$.  The frame starts at time $t_f^{(n)}$ and
has size $T_n[f]=1$. The active variable becomes $H_n(t_f^{(n)})=1$ and an activation cost of $e_n$ is incurred on this slot $t_f^{(n)}$. A random service variable $\mu_n(t_f^{(n)})$ is generated and requests are served according to the queue update \eqref{queue_update}. Recall that, under Assumption \ref{observability}, the  value of $\mu_n(t)$ is not known until the end of the slot.

\item Case $\alpha_n[f] \in \mathcal{L}_n$.  In this case, the server chooses to go idle and $\alpha_n[f]$ represents the specific sleep mode chosen.  The idle duration $I_n[f]$ is also chosen as an integer in the set $[1, I_{max}]$.  After the idle duration completes, the setup duration starts and has an independent and  random
duration $\tau_n[f] = \hat{\tau}(\alpha_n[f])$, where $\hat{\tau}(\alpha_n[f])$ is an integer random variable with a  known mean and variance that depends on the sleep mode $\alpha_n[f]$.
At the end of the setup time the system goes active and serves with a random $\mu_n(t)$ as before.
The active variable is $H_n(t)=0$ for all slots $t$ in the idle and setup times, and is $1$ at the very last slot of the frame.
Further:
\begin{itemize}
\item Idle cost:  Every slot $t$ of the idle time, an idle cost of $g_n(t) = \hat{g}_n(\alpha_n[f])$ is incurred (so that the idle cost depends on the sleep mode).  The process $g_n(t)=0$ if server $n$ is not idle on slot $\tau$.  The idle cost can be zero, but can also be a small but positive value if some electronics are still running in the sleep mode chosen.

\item Setup cost: Every slot $t$ of the setup time,
a cost of $W_n(t)=\hat{W}_n(\alpha_n[f])$ is incurred. The process $W_n(t)=0$ if server $n$ is not in a setup duration on slot $\tau$.
\end{itemize}
\end{itemize}

Thus, the length of frame $f$ for server $n$ is:
\begin{equation}\label{frame_length}
T_n[f]=\left\{
         \begin{array}{ll}
           1, & \hbox{if $\alpha_n[f]=active$;}\\
           I_n[f]+\tau_n[f]+1, & \hbox{if $\alpha_n[f]\in\mathcal{L}_n$.}
         \end{array}
       \right.
\end{equation}
In summary, the costs $\hat{g}_n(\alpha_n)$, $\hat{W}_n(\alpha_n)$ and the setup time $\hat{\tau}_n(\alpha_n)$ are functions of $\alpha_n\in\mathcal{L}_n$. We further make the following assumption regarding $\hat{\tau}_n(\alpha_n)$:

\begin{assumption}\label{bounded_moment_assumption}
For any $\alpha_n\in\mathcal{L}_n$, the function $\hat{\tau}_n(\alpha_n)$ is an integer random variable with known mean and variance, as well as bounded first four moments. Denote $\expect{\tau_n(\alpha_n)}=m_{\alpha_n}$ and $\textrm{Var}[\tau_n(\alpha_n)]=\sigma_{\alpha_n}^2$.
\end{assumption}
\revise{Note that this is a very mild assumption in view of the fact that the setup time of a real server is always bounded. The motivation behind emphasizing the fourth moment here instead of simply proceeding with boundedness assumption is more of theoretical interest than practical importance.
}

\revise{
Table I summarizes the parameters introduced in this section. The data center architecture is shown is Fig. \ref{fig:Stupendous0}. Since different servers might make different decisions, the renewal frames are not necessarily aligned. A typical asynchronous timeline for a system of three servers is illustrated in Fig. \ref{fig:multi-system}.
}

\begin{table}
\begin{center}
\caption{Parameters}
\begin{tabular}{|l|l|}
  \hline
  Control parameters & Control objectives \\
  \hline
  $R_n(t)$ & Requests routed to server $n$ at slot $t$ \\ 
  $d(t)$ & Requests rejected at slot $t$ \\
  $\alpha_n[f]$ & The option (active/idle) server $n$ takes in frame $f$ \\
  $I_n[f]$ & Number of slots server $n$ stays idle in frame $f$ \\
  \hline
  Other parameters & Meaning\\
  \hline
  $\lambda(t)$ & Number of arrivals at time $t$\\
  $c(t)$ & Per request rejection cost at time $t$\\
  $e_n$ & Per slot active service cost for server $n$\\
  $T_n[f]$ & The length of frame $f$ for server $n$\\
  $t^{(n)}[f]$ & Starting slot of frame $f$ for server $n$\\
  $\tau_n[f]$ & Setup duration in frame $f$\\
  $\mu_n(t)$ & Number of requests served on server $n$ at time $t$\\
  $H_n(t)$ & Server active indicator (equal to 1 if active, 0 if not)\\
  $g_n(t)$ & Idle cost of server $n$ at time $t$\\
  $W_n(t)$ & Setup cost of server $n$ at time $t$\\
  \hline
\end{tabular}
\end{center}
\end{table}
\begin{figure}[htbp]
   \centering
   \includegraphics[height=2in]{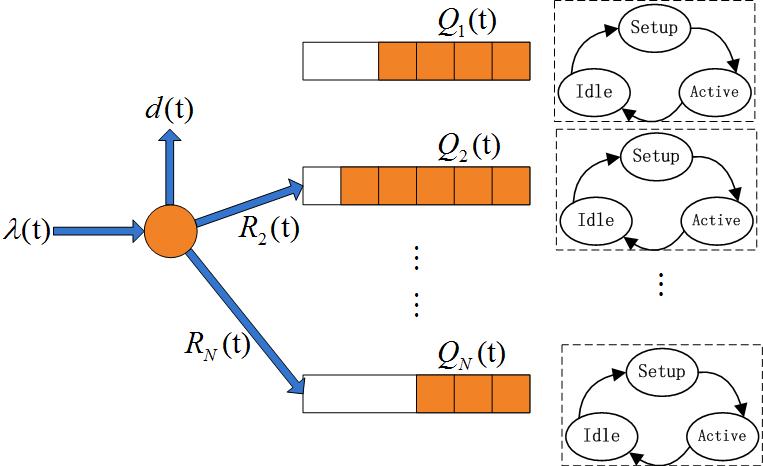} % requires the graphicx package
   \caption{\textbf{Illustration of a data center structure which contains a front-end load balancer, $N$ application servers with $N$ request queues and a backend database (omitted here for brevity)}.}
   \label{fig:Stupendous0}
\end{figure}

\begin{figure}[htbp]
   \centering
   \includegraphics[height=1in]{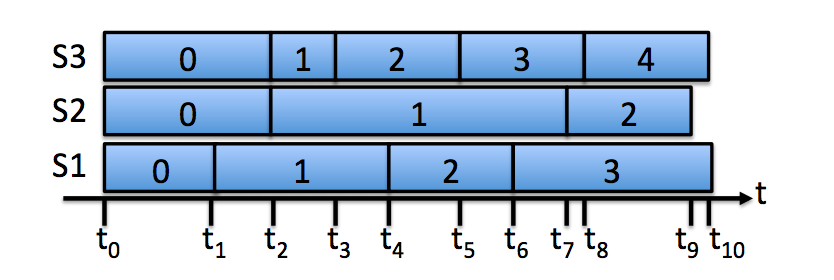} % requires the graphicx package
   \caption{Illustration of a typical asynchronous timeline for a system of three servers, where we count the number of renewal frames per single server.}
   \label{fig:multi-system}
\end{figure}

\subsection{Performance Objective}

For each $n \in \mathcal{N}$,
let $\overline{C}$, $\overline{W}_n$, $\overline{E}_n$, $\overline{G}_n$ be the time average costs resulting from rejection, setup, service and idle, respectively. They are defined as follows:
$\overline{C}=\lim_{T\rightarrow\infty}\frac{1}{T}\sum_{t=0}^{T-1}d(t)c(t)$,
$\overline{W}_n=\lim_{T\rightarrow\infty}\frac{1}{T}\sum_{t=0}^{T-1}W_n(t)$,
$\overline{E}_n=\lim_{T\rightarrow\infty}\frac{1}{T}\sum_{t=0}^{T-1}e_nH_n(t)$,
$\overline{G}_n=\lim_{T\rightarrow\infty}\frac{1}{T}\sum_{t=0}^{T-1}g_n(t)$.
%In the case when above limits are not guaranteed to exist, we replace $\lim$ with $\limsup$.

The goal is to design a joint routing and service policy so that the time average overall cost is minimized and all queues are stable, i.e.
\begin{align}\label{obj_1}
\min~\overline{C}+\sum_{n=1}^N\left(\overline{W}_n+\overline{E}_n+\overline{G}_n\right),~
\textrm{s.t.}~Q_n(t)~\textrm{stable } \forall n.
\end{align}
Notice that the constraint in \eqref{obj_1} is not easy to work with. In order to get an optimization problem one can deal with, we further define the time average request rate, rejection rate, routing rate and service rate as $\overline{\lambda}$, $\overline{d}$, $\overline{R}_n$, and $\overline{\mu}_n$ respectively:
$\overline{\lambda}=\lim_{T\rightarrow\infty}\frac{1}{T}\sum_{t=0}^{T-1}\lambda(t)$,
$\overline{d}=\lim_{T\rightarrow\infty}\frac{1}{T}\sum_{t=0}^{T-1}d(t)$,
$\overline{R}_n=\lim_{T\rightarrow\infty}\frac{1}{T}\sum_{t=0}^{T-1}R_n(t)$,
$\overline{\mu}_n=\lim_{T\rightarrow\infty}\frac{1}{T}\sum_{t=0}^{T-1}\mu_n(t)H_n(t)$.

Then, rewrite the problem \eqref{obj_1} as follows
\begin{align}
\min~~&\overline{C}+\sum_{n=1}^N\left(\overline{W}_n+\overline{E}_n+\overline{G}_n\right) \label{obj_3}\\
\textrm{s.t.}~~&\overline{R}_n\leq\overline{\mu}_n,~\forall n\in\mathcal{N}\label{obj_4}\\
&\sum_{n=1}^NR_n(t)\leq R_{\max},~\sum_{n=1}^NR_n(t)+d(t)=\lambda(t)~\forall t \label{obj_5}
\end{align}
Constraint \eqref{obj_4} and \eqref{obj_5} suggest that the time average request arrival rate is no greater than the time average total service rate.
We aim to develop an algorithm so that each server can make its own decision (without looking at the workload or service decision of any other server) and prove its near optimality.

\section{An online control algorithm}\label{section_proposed_algorithm}
We now present an online control algorithm which makes the joint request routing and service decisions as functions of $(\lambda(t), c(t), \mathbf{Q}(t))$, where $\mathbf{Q}(t)=(Q_1(t), \ldots, Q_N(t))$. The algorithm is inherited from the Lyapunov optimization technique and features a trade-off parameter $V>0$. 
\revise{In view of the fact that the traditional Lyapunov optimization technique does not require the statistics of the arrivals in a stochastic network optimization problem(see \cite{Lyapunov_optimization, Lyapunov_book} for details), our proposed algorithm does not need the statistics of $\lambda(t)$, $c(t)$ either and yields a simple distributed implementation.} 
%Specifically, each server makes frame-based decisions regarding its own queue state $Q_n(t)$, whereas the system makes routing and rejection decisions slot-wise based on $(\lambda(t), c(t), \mathbf{Q}(t))$.

\subsection{Intuition on the algorithm design}
Recall that our goal is to make routing and service decisions so as to solve the optimization problem \eqref{obj_3}-\eqref{obj_5}. First of all, from the queueing model described in the last section and Fig. \ref{fig:Stupendous0}, it is intuitive that an efficient algorithm would have
each server make decisions regarding its own queue state $Q_n(t)$, whereas the front-end load-balancer make routing and rejection decisions slot-wise based on the global information $(\lambda(t), c(t), \mathbf{Q}(t))$.

Next, to get an idea on what exactly the decision should be, 
by virtue of Lyapunov optimization, one would introduce a trade-off parameter $V>0$ and penalize the time average constraint \eqref{obj_4} via $\mathbf{Q}(t)$ to solve the following slotwise optimization problem
\begin{align}
\min~~&V\left(c(t)d(t)+\sum_{n=1}^N\left(W_n(t)+e_nH_n(t)+g_n(t)\right)\right) \label{obj-temp}\\
&+\sum_{n=1}^NQ_n(t)(R_n(t)-\mu_n(t))
~~\textrm{s.t.}~~\textrm{constraint}~\eqref{obj_5},\nonumber
\end{align}
which is naturally separable regarding the load-balancing decision ($d(t)$, $R_n(t)$), and the service decision ($W_n(t),~ H_n(t),~ g_n(t),~ \mu_n(t)$).
However, because of the existence of setup state (on which no decision could be made), the server does not have identical decision set every slot and furthermore, the decision set itself depends previous decisions. This poses a significant difficulty analyzing the above optimization \eqref{obj-temp}.
%\footnote{It would be very interesting if one can get a theoretical performance bound by simply doing \eqref{obj-temp} slotwise.}. 

In order to resolve this difficulty, we try to find the smallest ``identical time unit'' for each individual server in lieu of slots. This motivates the notion of renewal frame in the previous section (see Fig. \ref{fig:single-system}). Specifically, from Fig. \ref{fig:single-system} and the related renewal frame construction, at the starting slot of each renewal, the server faces the identical decision set (remain active or go to idle with certain slots) regardless of what previous decisions are. Following this idea, we modify \eqref{obj-temp} as follows:
\begin{itemize}
\item For the front-end load balancer, we observe $(\lambda(t), c(t), \mathbf{Q}(t))$ and solve $min~Vc(t)d(t)+\sum_{n=1}^NQ_n(t)R_n(t),~s.t. ~\eqref{obj_5}$, which is detailed in Section \ref{front-section}.
\item For each server, instead of per slot optimization $\min~V(W_n(t)+e_nH_n(t)+g_n(t))-Q_n(t)\mu_n(t)$, we propose to minimize the time average of this quantity per renewal frame $T_n[f]$, which is \eqref{server_decision} in Section \ref{server-section}.
\end{itemize}

\subsection{Thresholding algorithm on admission and routing}\label{front-section}
During each time slot, the data center chooses integers $R_n(t)$ and $d(t)$ to solve the following optimization problem:
\begin{align}
\min&~~Vd(t)c(t)+\sum_{n=1}^NQ_n(t)R_n(t) \label{front_end_routing}\\
\textrm{s.t.}&~~\sum_{n=1}^NR_n(t)+d(t)=\lambda(t),~0\leq \sum_{n=1}^NR_n(t)\leq R_{\max}. \nonumber
\end{align}
Notice that the solution of this problem admits a simple thresholding rule, i.e.
\begin{equation}\label{reject_decision}
d(t)=\left\{
       \begin{array}{lll}
         \max\{\lambda(t)-R_{\max},~0\},&\hbox{if $\exists n\in\mathcal{N}$ s.t.} \\
         &\hbox{ $Q_n(t)\leq Vc(t)$;}\\
         \lambda(t),&\hbox{otherwise.}
       \end{array}
     \right.
\end{equation}
\begin{equation}\label{routing_decision}
R_n(t)=\left\{
         \begin{array}{lll}
           \min\{\lambda(t), R_{\max}\}, & \hbox{if $Q_n(t)$ is the shortest } \\
                                         &\hbox{queue and $Q_n(t)\leq Vc(t)$;}\\
           0, & \hbox{otherwise.}
         \end{array}
       \right.
\end{equation}

\subsection{Ratio of expectation minimization on service decisions}\label{server-section}
Each server makes its own frame-based service decisions. Specifically, for server $n$, at the beginning of its $f$-th renewal frame $t_f^{(n)}$, it observes its current queue state $Q(t_f^{(n)})$ and makes decisions on $\alpha_n[f]\in\{active\}\cup\mathcal{L}_n$ and $I_n[f]$ so as to solve the minimization of ratio of expectations in \eqref{server_decision},
where $B_0=\frac{1}{2}(R_{\max}+\mu_{\max})\mu_{\max}$. Recall the definition of $T_n[f]$, if the server chooses to remain active, then the frame length is exactly 1, otherwise, the server is allowed to choose how long it stays in idle. It can be easily shown that over all randomized decisions between staying active and going to different idle states, it is optimal to make a pure decision which either stays active or goes to one of the idle states with probability 1.
Thus, we are able to simplify the problem by computing $D_n[f]$ for active and idle options separately.
\begin{itemize}
  \item If the server chooses to go active, i.e. $\alpha_n[f]=active$, then,
\begin{equation}\label{DPP_active}
D_n[f]=Ve_n-Q_n(t_f^{(n)})\mu_n.
\end{equation}
  \item If the server chooses to go idle, i.e. $\alpha_n[f]\in\mathcal{L}_n$, then, $D_n[f]$ is shown in \eqref{DPP_idle_origin},
which follows from the fact that if the server goes idle, then, $H_n(t)$ are all zero during the frame except for the last slot. Now we try to compute the optimal idle option $\alpha_n[f]\in\mathcal{L}_n$ and idle time length $I_n[f]$ given the server chooses to go idle. The following lemma illustrates that the decision on $I_n[f]$ can also be reduced to pure decision.
\begin{lemma}\label{compute_idle}
The best decision minimizing \eqref{DPP_idle_origin} is a pure decision which takes one $\alpha_n[f]\in\mathcal{L}_n$ and one integer value $I_n[f]\in\left\{1,\cdots,I_{\max}\right\}$ minimizing the deterministic function \eqref{DPP_idle}.
\end{lemma}
The proof of above lemma is given in appendix A.
\end{itemize}
Then, the server computes the minimum of \eqref{DPP_idle}, which is nothing but a deterministic optimization problem. It
 goes in the following two steps:
\begin{enumerate}
  \item For each $\alpha_n\in\mathcal{L}_n$, compute optimal
$I_n[f]$, which is one of the two integers closest to \eqref{solution_idle}
\begin{figure*}
\normalsize
\begin{equation}\label{server_decision}
D_n[f]\triangleq
\frac{\mathbb{E}\left[\left.\sum_{t=t_f^{(n)}}^{t=t_{f+1}^{(n)}-1}\left(VW_n(t)+Ve_nH_n(t)+Vg_n(t)-Q_n(t_f^{(n)})\mu_n(t)H_n(t)\right)
+\left(t-t_f^{(n)}\right)B_0~\right|~Q_n(t_f^{(n)})\right]}{\expect{T_n[f]~\left|~Q_n(t_f^{(n)})\right.}}
\end{equation}
\begin{align}\label{DPP_idle_origin}
D_n[f]=\frac{V\hat{W}_n(\alpha_n[f])m_{\alpha_n[f]}+Ve_n-Q_n(t_f^{(n)})\mu_n+\expect{V\hat{g}(\alpha_n[f])I_n[f]+
  \frac{B_0}{2}T_n[f](T_n[f]-1)\left|~Q_n(t_f^{(n)})\right.}}
{\expect{T_n[f]~\left|~Q_n(t_f^{(n)})\right.}}
\end{align}
\begin{equation}\label{DPP_idle}
D_n[f]=\frac{V\hat{W}_n(\alpha_n[f])m_{\alpha_n[f]}+Ve_n-Q_n(t_f^{(n)})\mu_n+\frac{B_0}{2}\sigma_n^2+V\hat{g}(\alpha_n[f])I_n[f]}
{I_n[f]+m_{\alpha_n[f]}+1}
+\frac{B_0}{2}(I_n[f]+m_{\alpha_n[f]}+1).
\end{equation}
\begin{equation}\label{solution_idle}
\left\lceil\sqrt{\left\lfloor\frac{B_0}{2}\left( Vm_{\alpha_n}\left(\hat{W}_n(\alpha_n)-g_n(\alpha_n)\right)
  +V\left(e_n-g_n(\alpha_n)\right)-U_n(t_f^{(n)})\mu_n+\frac{B_0}{2}\sigma_n^2\right)\right\rfloor_{(m_{\alpha_n}+1)^2}}
-m_{\alpha_n}-1\right\rceil^{I_{\max}}
\end{equation}
\end{figure*}
which achieves smaller value on \eqref{DPP_idle},
where $\lfloor\cdot\rfloor$ and$\lceil\cdot\rceil$ stand for floor function and ceiling function respectively.
  \item Compare \eqref{DPP_idle} for different $\alpha_n\in\mathcal{L}_n$ and choose the one achieving the minimum.
\end{enumerate}
Then, the server compares \eqref{DPP_active} with the minimum of \eqref{DPP_idle}. If \eqref{DPP_active} is less than the minimum of \eqref{DPP_idle}, then, the server chooses to go active. Otherwise, the server chooses to go idle and stay idle for the time length given the one achieving the minimum of \eqref{DPP_idle}.

Above all, our algorithm is summarized in Algorithm 1.
\begin{algorithm}
\begin{Alg}~
\begin{itemize}
  \item At each time slot $t$, the data center observes $\lambda(t)$, $c(t)$, and $\mathbf{Q}(t)$ chooses rejection decision $d(t)$ according to \eqref{reject_decision} and chooses routing decision $R_n(t)$ according to \eqref{routing_decision}.
  \item For each server $n\in\mathcal{N}$, at the beginning of its $f$-th frame $t_f^{(n)}$, observe its queue state $Q_n(t_f^{(n)})$ and compute \eqref{DPP_active} and the minimum of \eqref{DPP_idle}. If \eqref{DPP_active} is less than the minimum of \eqref{DPP_idle}, then the server still stays active. Otherwise, the server switches to the idle state minimizing \eqref{DPP_idle} and stays idle for $I_n[f]$ achieving the minimum of \eqref{DPP_idle}.
  \item Update $Q_n(t),~\forall n\in\mathcal{N}$ according to \eqref{queue_update}.
\end{itemize}
\end{Alg}
\end{algorithm}

\section{Performance Analysis}\label{section_performance_analysis}
\revise{In this section, we prove the online algorithm introduced in the last section makes all request queues $Q_n(t)$ bounded (on the order of $V$) and achieves the near optimality with sub-optimality gap on the order of $1/V$.}

\subsection{Bounded request queues}
\begin{lemma}\label{bounded_delay}
If $Q_n(0)=0,~\forall n\in\mathcal{N}$, then, each request queue $Q_n(t)$ is deterministically bounded with bound:
$Q_n(t)\leq Vc_{\max}+R_{\max},~\forall t,~\forall n\in\mathcal{N}$,
where $c_{\max}\triangleq \max_{c\in\mathcal{C}}c$.
\end{lemma}
\begin{proof}
We use induction to prove the claim. Base case is trivial since $Q_n(0)=0\leq Vc_{\max}+R_{\max}$. Suppose the claim holds at the beginning of $t=i$ for $i>0$, so that
$Q_n(i)\leq Vc_{\max}+R_{\max}.$
Then,
\begin{enumerate}
  \item If $Q_n(i)\leq Vc_{\max}$, then, it is possible for the queue to increase during slot $i$. However, the increase of the queue within one slot is bounded by $R_{\max}$. which implies at the beginning of slot $i+1$,
      $Q_n(i+1)\leq Vc_{\max}+R_{\max}.$
  \item If $Vc_{\max} < Q_n(i)\leq Vc_{\max}+R_{\max}$, then, according to \eqref{routing_decision}, it is impossible to route any request to server $n$ during slot $i$, and $R_n(i)=0$ which results in
      $Q_n(i+1)\leq Vc_{\max}+R_{\max}.$
\end{enumerate}
Above all, we finished the proof of lemma.
\end{proof}

\begin{lemma}\label{constraint_satisfy}
The proposed algorithm meets the constraint \eqref{obj_4} with probability 1.
\end{lemma}
\begin{proof}
From the queue update rule \eqref{queue_update}, it follows,
$Q_n(t+1)\geq Q_n(t)+R_n(t)-\mu_nH_n(t)$.
Taking telescoping sums from 0 to $T-1$ gives
$Q_n(T)\geq Q_n(0)+\sum_{t=0}^{T-1}R_n(t)-\sum_{t=0}^{T-1}\mu_nH_n(t)$.
Since $Q_n(0)=0$, dividing both sides by $T$ gives
$\frac{Q_n(T)}{T}\geq \frac{1}{T}\sum_{t=0}^{T-1}R_n(t)-\frac{1}{T}\sum_{t=0}^{T-1}\mu_nH_n(t)$.
Substitute the bound $Q_n(T)\leq Vc_{\max}+R_{\max}$ from lemma \ref{bounded_delay} into above inequality and take limit as $T\rightarrow\infty$ give the desired result.
\end{proof}
\subsection{Optimal randomized stationary policy}\label{op_stat_algorithm}
In this section, we introduce a class of algorithms which are theoretically helpful for doing analysis, but practically impossible to implement.

Since servers are coupled only through time average constraint \eqref{obj_4}, each server $n$ can be viewed as a separate renewal system, thus, it can be shown that any possible time average service rate $\overline{\mu}_n$ can be achieved through a frame based stationary randomized service decision, meaning that the decisions are i.i.d. over frames (the proof is similar to Lemma 4.2 of \cite{asyn-paper} and is omitted here for brevity).
Furthermore, it can be shown that the optimality of \eqref{obj_3}-\eqref{obj_5} can be achieved over the following randomized stationary algorithms: At the beginning of each time slot $t$, the data center observes the incoming requests $\lambda(t)$ and rejecting cost $c(t)$, then routes $R_n^*(t)$ incoming requests to server $n$ and rejects $d^*(t)$ requests, both of which are random functions on $(\lambda(t),c(t))$. They satisfy the same instantaneous relation as \eqref{obj_5}.
Meanwhile, server $n$ chooses a frame based stationary randomized service decision $(\alpha_n^*[f],I_n^*[f])$, so that the optimal service rate is achieved.

If one knows the stationary distribution for $(\lambda(t),c(t))$, then, this optimal control algorithm can be computed using dynamic programming or linear programming. Moreover, the optimal setup cost $W_n^*(t)$, idle cost $g_n^*(t)$, and the active state indicator $H^*(t)$ can also be deduced.
Since the algorithm is stationary, these three cost processes are all ergodic markov processes. Let $T_n^*[f]$ be the frame length process under this algorithm. Thus, it follows from renewal reward theorem that
$\left\{\sum_{t=t^{(n)}_f}^{t^{(n)}_{f+1}-1}W_n^*(t)\right\}_{f=0}^{+\infty}$,
$\left\{\sum_{t=t^{(n)}_f}^{t^{(n)}_{f+1}-1}g_n^*(t)\right\}_{f=0}^{+\infty}$,
$\left\{\sum_{t=t^{(n)}_f}^{t^{(n)}_{f+1}-1}e_nH_n^*(t)\right\}_{f=0}^{+\infty}$, $\left\{\sum_{t=t^{(n)}_f}^{t^{(n)}_{f+1}-1}\mu_n(t)H_n^*(t)\right\}_{f=0}^{+\infty}$
and
$\left\{T_n^*[f]\right\}_{f=0}^{+\infty}$
are all \emph{i.i.d. random variables over frames}. Let $\overline{C}^*$, $\overline{W}_n^*$, $\overline{G}_n^*$ and $\overline{E}_n^*$ be the optimal time average costs. Let $\overline{R}_n^*$, $\overline{\mu}_n^*$ and $\overline{d}^*$ be the optimal time average routing rate, service rate and rejection rate respectively. Then, by strong law of large numbers,
\begin{equation}
\overline{W}_n^*
=\frac{\expect{\sum_{t=t^{(n)}_f}^{t^{(n)}_{f}+T_n^*[f]-1}W_n^*(t)}}{\expect{T_n^*[f]}} \label{iid_W}
\end{equation}
\begin{equation}
\overline{E}_n^*
=\frac{\expect{\sum_{t=t^{(n)}_f}^{t^{(n)}_{f}+T_n^*[f]-1}e_nH_n^*(t)}}{\expect{T_n^*[f]}} \label{iid_E}
\end{equation}
\begin{equation}
\overline{G}_n^*
=\frac{\expect{\sum_{t=t^{(n)}_f}^{t^{(n)}_{f}+T_n^*[f]-1}g_n^*(t)}}{\expect{T_n^*[f]}} \label{iid_G}
\end{equation}
\begin{align}
\overline{\mu}_n^*
=\frac{\expect{\sum_{t=t^{(n)}_f}^{t^{(n)}_{f}+T_n^*[f]-1}\mu_n(t)H_n^*(t)}}{\expect{T_n^*[f]}}, \label{iid_mu}
\end{align}
Also, notice that $R_n^*(t)$ and $d^*(t)$ depends only on the random variable $\lambda(t)$ and $c(t)$, which is i.i.d. over slots. Thus, $R_n^*(t)$ and $d^*(t)$ are also \emph{i.i.d. random variables over slots}. By law of large numbers,
\begin{align}
\overline{R}_n^*=&\expect{R_n^*(t)}, \label{iid_R}\\
\overline{C}^*=&\expect{c(t)d^*(t)}. \label{iid_d}
\end{align}

\begin{remark}
Notice that the idle time $I_n^*[f]\in[1,I_{\max}]$ and the first two moments of the setup time are bounded, it follows the first two moments of $T_n^*[f]$ are bounded.
\end{remark}

%Although the optimal stationary algorithm exists, it is practically impossible to compute because the data center has no knowledge about stationary distribution for $(\lambda(t),c(t))$.
%However, as we shall see in the next section, this algorithm provides us with an important baseline performance guarantee where our analysis relies on.

\subsection{Key features of thresholding algorithm}
In this part, we compare the algorithm deduced from the two optimization problems \eqref{front_end_routing} and \eqref{server_decision} to that of the best stationary algorithm in section \ref{op_stat_algorithm}, illustrating the key features of the proposed online algorithm.
Define $\mathcal{F}(t)$ as the system history up till slot $t$, which includes all the decisions taken and all the random events before slot $t$. Let's first consider \eqref{front_end_routing}. For simplicity of notations, define two random processes $\{X_n[f]\}_{f=0}^{\infty}$ and $\{Z[t]\}_{t=0}^{\infty}$ as follows
\begin{align*}
X_n[f]=&\sum_{t=t_f^{(n)}}^{t=t_{f+1}^{(n)}-1}\left(V\left(W_n(t)-\overline{W}_n^*\right)+V\left(e_nH_n(t)-\overline{E}_n^*\right)\right.\\
       &+V\left(g_n(t)-\overline{G}_n^*\right)-Q_n(t_f^{(n)})\left(\mu_nH_n(t)-\overline{\mu}^*\right)\\
       &\left.+\left(t-t_f^{(n)}\right)B_0-\Psi_n\right),\\
Z[t]=&V\left(c(t)d(t)-\overline{C}^*\right)+\sum_{n=1}^NQ_n(t)\left(R_n(t)-\overline{R}_n^*\right),
\end{align*}
where $\Psi_n=\frac{B_0}{2}\frac{\expect{T^*_n[f](T^*_n[f]-1)}}{\expect{T^*_n[f]}}$ and $B_0=\frac{1}{2}(R_{\max}+\mu_{\max})\mu_{\max}$.

Given the system info $\mathcal{F}(t)$, the random events $c(t)$ and $\lambda(t)$, the solutions \eqref{reject_decision} and \eqref{routing_decision} take rejecting and routing decisions so as to minimize \eqref{front_end_routing} over all possible routing and rejecting decisions at time slot $t$. Thus, the proposed algorithm achieves smaller value on \eqref{front_end_routing} compared to that of the best stationary algorithm in section \ref{op_stat_algorithm}. Formally, this idea can be stated as the following inequality:
$\expect{\left.Vc(t)d(t)+\sum_{n=1}^NQ_n(t)R_n(t)~\right|~c(t),\lambda(t),\mathcal{F}(t)}$
$\leq \expect{\left.Vc(t)d^*(t)+\sum_{n=1}^NQ_n(t)R_n^*(t)~\right|~c(t),\lambda(t),\mathcal{F}(t)}$.
Taking expectation regarding $c(t)$ and $\lambda(t)$
using the fact that the best stationary algorithm on $R^*_n(t)$ and $d^*(t)$ are i.i.d. over slots (independent of $\mathcal{F}(t)$), together with \eqref{iid_R} and \eqref{iid_d}, we get
\begin{equation}\label{front_end_feature}
\expect{\left.Z(t)~\right|~\mathcal{F}(t)}\leq 0.
\end{equation}
Similarly, for \eqref{server_decision}, the proposed service decisions within frame $f$ minimizes $D_n[f]$ in \eqref{server_decision}, thus, compared to that of the best stationary policy, \eqref{op_stat_interim} holds.
\begin{figure*}
\begin{align}\label{op_stat_interim}
&\frac{\mathbb{E}\left[\left.\sum_{t=t_f^{(n)}}^{t=t_{f+1}^{(n)}-1}\left(V(W_n(t)+e_nH_n(t)+g_n(t))
        -Q_n(t_f^{(n)})\mu_n(t)H_n(t)\right)+\left(t-t_f^{(n)}\right)B_0~\right|~\mathcal{F}(t_f^{(n)})\right]}
{\expect{T_n[f]~\left|~\mathcal{F}(t_f^{(n)})\right.}}\nonumber\\
\leq&
\frac{\expect{\left.\sum_{t=t_f^{(n)}}^{t=t_{f}^{(n)}+T_n^*[f]-1}\left(V\left(W_n^*(t)+e_nH_n^*(t)+g_n^*(t)\right)-Q_n(t_f^{(n)})\mu_nH_n^*(t)\right)
+\frac{B_0}{2}T^*_n[f](T^*_n[f]-1)~\right|~\mathcal{F}(t_f^{(n)})}}
{\expect{T^*_n[f]~\left|~\mathcal{F}(t_f^{(n)})\right.}}
\end{align}
\end{figure*}
Again, using the fact that the optimal stationary algorithm gives i.i.d. $W_n^*(t)$, $g_n^*(t)$, $H_n^*(t)$ and $T_n^*[f]$ over frames (independent of $\mathcal{F}(t_f^{(n)})$),  as well as \eqref{iid_W}, \eqref{iid_E} and \eqref{iid_mu}, we get
\begin{align}
\expect{X_n[f]~\left|~\mathcal{F}(t_f^{(n)})\right.}
\left/\expect{T_n[f]~\left|~\mathcal{F}(t_f^{(n)})\right.}\right.\leq0 \label{server_feature}
\end{align}

\subsection{Bounded average of supermartingale difference sequeces}
The key feature inequalities \eqref{front_end_feature} and \eqref{server_feature} provide us with bounds on the expectations. The following lemma serves as a stepping stone passing from expectation bounds to probability 1 bounds. \revise{We need the following basic definition of supermartingale to start with:
\begin{definition}[Supermartingale]
Let $\{\mathcal{F}_t\}_{t=0}^{\infty}$ be a filtration, i.e. an increasing sequence of $\sigma$-fields. Let $\{Y_t\}_{t=0}^\infty$ be a sequence of random variables. Then, $\{Y_t\}_{t=0}^\infty$ is said to be a supermartingale if (i) $\expect{|Y_t|}<\infty,~\forall t$, (ii) $Y_t$ is measurable with respect to $\mathcal{F}_t$ for all $t$, (iii) $\expect{Y_{t+1}|\mathcal{F}_t}\leq Y_t,~\forall t$. The corresponding difference process $X_t=Y_{t+1}-Y_t$ is called the supermartingale difference sequence.
\end{definition}
}
\revise{
We also need the following strong law of large numbers for supermartingale difference sequences:
\begin{lemma}[Corollary 4.2 of \cite{prob_1_convergence}]\label{prob-1-converge}
Let $\{X_t\}_{t=0}^\infty$ be a supermartingale difference sequence. If 
$$\sum_{t=1}^\infty \left.\expect{X_t^2}\right/t^2<\infty,$$
then,
\[\limsup_{T\rightarrow\infty}\frac1T\sum_{t=0}^{T-1}X_t\leq 0,\]
with probability 1.
\end{lemma}
}

With this lemma, we are ready to prove the following result:
\begin{lemma}\label{bounded_supMG}
Under the proposed algorithm, the following hold with probability 1,
\begin{align}
&\limsup_{F\rightarrow\infty}\frac{1}{F}\sum_{f=0}^{F-1}X_n[f]\leq0, \label{prob_1_server}\\
&\limsup_{T\rightarrow\infty}\frac{1}{T}\sum_{t=0}^{T-1}Z[t]\leq0. \label{prob_1_front_end}
\end{align}
\end{lemma}
\begin{proof}
The key to the proof is treating these two sequences as super-martingale difference sequence and applying law of large numbers for super-martingale difference sequence (theorem 4.1 and corollary 4.2 in \eqref{prob_1_server}).

We first look at the sequence $\{X_n[f]\}_{f=0}^{\infty}$. Let
$Y_n[F]=\sum_{f=0}^{F-1}X_n[f]$.
We first prove that $Y_n[F]$ is a supermartingale. Notice that $Y_n[F]\in\mathcal{F}\left(t_n^{(F)}\right)$, i.e. it is measurable given all the information before frame $t_n^{(F)}$, and $|Y_n[F]|<\infty,~\forall F<\infty$. Furthermore,
$\expect{Y_n[F+1]-Y_n[F]~\left|~\mathcal{F}\left(t_F^{(n)}\right)\right.}$
$=\expect{X_n[F]~\left|~\mathcal{F}\left(t_F^{(n)}\right)\right.}$
%$=(\expect{X_n[F]~\left|~\mathcal{F}\left(t_F^{(n)}\right)\right.}\left/\expect{T_n[F]~\left|~\mathcal{F}\left(t_F^{(n)}\right)\right.}\right.)
%\cdot\expect{T_n[F]~\left|~\mathcal{F}\left(t_F^{(n)}\right)\right.}$
$\leq 0\cdot\expect{T_n[F]~\left|~\mathcal{F}\left(t_F^{(n)}\right)\right.}=0$,
where the only inequality follows from \eqref{server_feature}. Thus, it follows $Y_n[F]$ is a supermartingale. Next, we show that the second moment of supermartingale differences, i.e. $\expect{X_n[f]^2}$, is deterministically bounded by a fixed constant for any $f$. This part of proof is given in Appendix B. Thus, the following holds:
$\sum_{f=1}^{\infty}\expect{X_n[f]^2}\left/f^2\right.<\infty$.
Now, applying Lemma \ref{prob-1-converge} immediately gives \eqref{prob_1_server}.

Similarly, we can prove \eqref{prob_1_front_end} by proving $M[t]=\sum_{t=0}^{T-1}Z[t]$ is a supermartingale with bounded second moment on differences using \eqref{iid_R}, \eqref{iid_d} and \eqref{front_end_feature}. The procedure is almost the same as above and we omitted the details here for brevity.
\end{proof}

\begin{corollary}\label{corollary_ratio_time_average}
The following ratio of time averages is upper bounded with probability 1, \\
$\limsup_{F\rightarrow\infty}\left.\sum_{f=0}^{F-1}X_n[f]\right/\sum_{f=0}^{F-1}T_n[f]\leq 0$.
\end{corollary}
\begin{proof}
From \eqref{prob_1_server}, it follows for any $\epsilon>0$, there exists an $F_0(\epsilon)$ such that $F\geq F_0(\epsilon)$ implies
$\sum_{f=0}^{F-1}X_n[f]\left/\sum_{f=0}^{F-1}T_n[f]\right.\leq\epsilon\left/\frac{1}{F}\sum_{f=0}^{F-1}T_n[f]\right.\leq\epsilon$.
Thus,
$\limsup_{F\rightarrow\infty}\sum_{f=0}^{F-1}X_n[f]\left/\sum_{f=0}^{F-1}T_n[f]\right.\leq\epsilon$.
Since $\epsilon$ is arbitrary, take $\epsilon\rightarrow0$ gives the result.
\end{proof}

\subsection{Near optimal time average cost}
The ratio of time averages in corollary \ref{corollary_ratio_time_average} and the true time average share the same bound, which is proved by the following lemma:
\begin{lemma}\label{true_time_average}
The following time average is bounded,
\begin{align}
&\limsup_{T\rightarrow\infty}\frac{1}{T}\sum_{t=0}^{T-1}\left(V\left(W_n(t)+e_nH_n(t)+g_n(t)\right)-Q_n(t_f^{(n)})\right.\nonumber\\
&\left.(\mu_nH_n(t)-\overline{\mu}_n^*)+\left(t-t_f^{(n)}\right)B_0\right) \nonumber\\
&\leq V(\overline{W}_n^*+\overline{E}_n^*+\overline{G}_n^*)+\Psi_n,  \label{true_time_average_equation}
\end{align}
where $\Psi_n=\frac{B_0}{2}\frac{\expect{T^*_n[f](T^*_n[f]-1)}}{\expect{T^*_n[f]}}$ and $B_0=\frac{1}{2}(R_{\max}+\mu_{\max})\mu_{\max}$.
\end{lemma}

The idea of the proof is similar to that of basic renewal theory, which derives upper and lower bounds for each $T$ within any frame $F$ using corollary \ref{corollary_ratio_time_average}, thereby showing that as $T\rightarrow\infty$, the upper and lower bounds meet. See appendix C for details. With the help of this lemma, we are able to prove the following near optimal performance theorem:

\begin{theorem}\label{theorem_near_optimal_perform}
If $Q_n(0)=0,\forall n\in\mathcal{N}$, then the time average total cost under the algorithm is near optimal on the order of $\mathcal{O}(1/V)$, i.e.
\begin{align}\label{near_optimal_perform}
&\limsup_{T\rightarrow\infty}\frac{1}{T}\sum_{t=0}^{T-1}\left(c(t)d(t)+\sum_{n=1}^{N}\left(W_n(t)+e_nH_n(t)+g_n(t)\right)\right)\nonumber\\
&\leq\underbrace{\overline{C}^*+\sum_{n=1}^N\left(\overline{W}_n^*+\overline{E}_n^*+\overline{G}_n^*\right)}_\text{Optimal cost}+\frac{\sum_{n=1}^N\Psi_n+B_3}{V},
\end{align}
where $B_3\triangleq\frac{1}{2}\sum_{n=1}^N(R_{\max}+\mu_n)^2$, $\Psi_n=\frac{B_0}{2}\frac{\expect{T^*_n[f](T^*_n[f]-1)}}{\expect{T^*_n[f]}}$ and $B_0=\frac{1}{2}(R_{\max}+\mu_{\max})\mu_{\max}$.
\end{theorem}
The proof uses Lyapunov drift analysis. See appendix D for details of proof.

\section{Delay improvement via virtualization}\label{section_improve}
\subsection{Delay improvement}
The algorithm in previous sections optimizes time average cost. However, it can route requests to idle queues, which increases system delay. This section considers an improvement in the algorithm that maintains the same average cost guarantees, but reduces delay. This is done by a ``virtualization'' technique that reduces from $N$ server request queues to only one request queue $Q(t)$. Specifically, the same Algorithm 1 is run, with queue updates \eqref{queue_update} for each of the $N$ queues $Q_n(t)$. However, the $Q_n(t)$ processes are now virtual queues rather than actual queues: Their values are only kept in software. Every slot $t$, the data center observes the incoming requests $\lambda(t)$, rejection cost $c(t)$ and virtual queue values, making rejection decision according to \eqref{reject_decision} as before. The admitted requests are queued in $Q(t)$. Meanwhile, each server $n$ makes active/idle decisions observing its own virtual queue $Q_n(t)$ same as before. Whenever a server is active, it grabs the requests from request queue $Q(t)$ and serves them.
This results in an actual queue updating for the system:
\begin{equation}\label{actual_queue_update}
Q(t+1)=\max\left\{Q(t)+\lambda(t)-d(t)-\sum_{n=1}^N\mu_n(t)H_n(t),~0\right\}.
\end{equation}

Fig. \ref{fig:Stupendous1} shows this data center architecture.
\begin{figure}[htbp]
   \centering
   \includegraphics[height=2in]{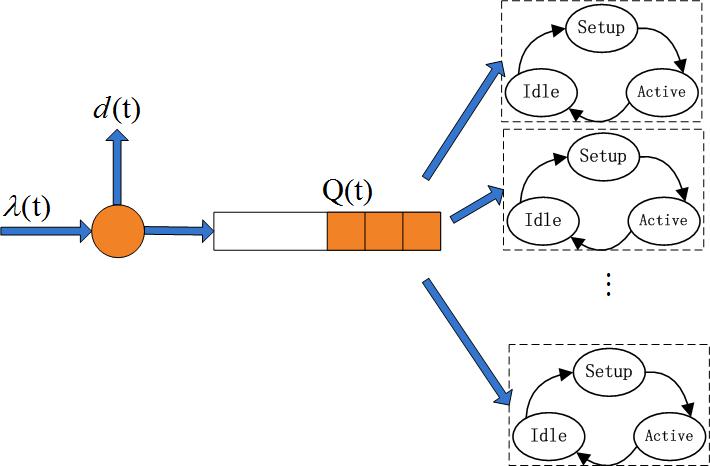} % requires the graphicx package
   \caption{Illustration of basic data center architecture.}
   \label{fig:Stupendous1}
\end{figure}

\subsection{Performance guarantee}
Since this algorithm does not look at the actual queue $Q(t)$, it is not clear whether or not the actual request queue would be stabilized under the proposed algorithm. The following lemma answers the question. For simplicity, we call the system with $N$ queues, where our algorithm applies, the virtual system, and call the system with only one queue the actual system.
\begin{lemma}\label{actual_queue_bound}
If $Q(0)=0$ and $Q_n(0)=0,~\forall n\in\mathcal{N}$, then the virtualization technique stabilizes the queue $Q(t)$ with the bound:
$Q(t)\leq N(Vc_{\max}+R_{\max})$.
\end{lemma}
\begin{proof}
Notice that this bound is $N$ times the individual queue bound in lemma \ref{bounded_delay}, we prove the lemma by showing that the sum-up weights $\sum_{n=1}^NQ_n(t)$ in the virtual system always dominates the queue length $Q(t)$. We prove this by induction. The base case is obvious since $Q(0)=\sum_{n=1}^NQ_n(0)=0$. Suppose at the beginning of time $t$,
$Q(t)\leq\sum_{n=1}^NQ_n(t)$,
then, during time $t$, we distinguish between the following two cases:
\begin{enumerate}
  \item Not all active servers in actual system have requests to serve. This case happens if and only if there are not enough requests in $Q(t)$ to be served, i.e.
      $\lambda(t)-d(t)+Q(t) < \sum_{n=1}^N\mu_n(t)H_n(t)$.
      Thus, according to queue updating rule \eqref{actual_queue_update}, at the beginning of time slot $t+1$, there will be no request sitting in the actual queue, i.e. $Q(t+1)=0$. Hence, it is guaranteed that
      $Q(t+1)\leq\sum_{n=1}^NQ_n(t+1)$.
  \item All active servers in actual system have requests to serve. Notice that the virtual system and the actual system have exactly the same arrivals, rejections and server active/idle states. Thus, the following holds,
      $Q(t+1)=Q(t)+\lambda(t)-d(t)-\sum_{n=1}^N\mu_n(t)H_n(t)$
            $\leq \sum_{n=1}^NQ_n(t)+\sum_{n=1}^NR_n(t)-\sum_{n=1}^N\mu_n(t)H_n(t)$
            $\leq \sum_{n=1}^N\max\{Q_n(t)+R_n(t)-\mu_n(t)H_n(t),~0\}$
            $= \sum_{n=1}^NQ_n(t+1)$,
      where the first inequality follows from induction hypothesis as well as the fact that $\sum_{n=1}^NR_n(t)=\lambda(t)-d(t)$.
\end{enumerate}
Above all, we proved $Q(t)\leq\sum_{n=1}^NQ_n(t),~\forall t$. Since each $Q_n(t)\leq Vc_{\max}+R_{\max},~\forall t$, the lemma follows.
\end{proof}

Since the virtual system and the actual system have exactly the same cost, and it can be shown that the optimal cost in one queue system is lower bounded by the optimal cost in $N$ queue system, thus, the near optimal performance is still guaranteed.

\section{Simulation}\label{section_simulation}
\revise{
In this section, we demonstrate the performance of our proposed algorithm via extensive simulations. The first simulation runs over i.i.d. traffic which fits into the assumption made in Section \ref{section_problem_formulation}. We show that our algorithm indeed achieves $\mathcal{O}(1/V)$ near optimality with $\mathcal{O}(V)$ delay ($[\mathcal{O}(1/V), \mathcal{O}(V)]$ trade-off), which is predicted by Lemma \ref{bounded_delay} and Theorem \ref{theorem_near_optimal_perform}. We then apply our algorithm to a real data center traffic trace with realistic scale, setup time and cost being the power consumption. We compare the performance of the proposed algorithm with several other heuristic algorithms and show that our algorithm indeed delivers lower delay and saves power.
}

\subsection{Near optimality in $N$ queues system}
In the first simulation, we consider a relative small scale problem with i.i.d. generated traffic.
We set the number of servers $N=5$. The incoming requests $\lambda(t)$ are integers following a uniform distribution in $[10,30]$. The request rejecting cost $c(t)$ are also integers following a uniform distribution in $[1,6]$. The maximum admission amount $R_{\max}=40$ and the maximum idle time $I_{\max}=1000$.
There is only one idle option $\alpha_n$ for each server where the idle cost $\hat{g}(\alpha_n)=0$.
The setup time follows a geometric distribution with mean $\expect{\hat{\tau}(\alpha_n)}$, setup cost $\hat{W}_n(\alpha_n)$ per slot, service cost $e_n$ per slot, and the service amount $\mu_n$ follows a uniform distribution over integers. The values $1/\expect{\hat{\tau}(\alpha_n)}$ are generated uniform at random within $[0,1]$ and specified in table II.

The algorithm is run for 1 million slots in each trial and each plot takes the average of these 1 million slots. We compare our algorithm to the optimal stationary algorithm. The optimal stationary algorithm is computed using linear program \cite{LP_MDP} with the full knowledge of the statistics of requests and rejecting costs.

\begin{table}
\begin{center}
\caption{Problem parameters}
\begin{tabular}{|l|l|l|l|l|}
  \hline
   Server & $\mu_n$ & $e_n$ & $\hat{W}_n(\alpha_n)$ & $\expect{\hat{\tau}(\alpha_n)}$ \\
  1 & $\{2,3,4,5,6\}$ & 4 & 2 & 5.893 \\
  2 & $\{2,3,4\}$ & 2 & 3 & 4.342 \\
  3 & $\{2,3,4\}$ & 3 & 3 & 27.397 \\
  4 & $\{1,2,3\}$ & 4 & 2 & 5.817 \\
  5 & $\{2,3,4\}$ & 2 & 4 & 6.211 \\
  \hline
\end{tabular}
\end{center}
\end{table}

In Fig. \ref{fig:Stupendous2}, we show that as our tradeoff parameter $V$ gets
larger, the average cost approaches the optimal value and
achieves a near optimal performance. \revise{Furthermore, the cost curve drops rapidly when $V$ is small and becomes relatively flat when $V$ gets large, thereby demonstrates our $\mathcal{O}(1/V)$ optimality gap in Theorem \ref{theorem_near_optimal_perform}.} 
Fig. \ref{fig:Stupendous3} plots the average sum-up queue size $\sum_{n=1}^5Q_n(t)$
and shows as $V$ gets larger, the average sum-up queue size becomes larger. We also plot the sum of individual queue bound
from Lemma \ref{bounded_delay} for comparison. \revise{We can see that the real queue size grows linearly with $V$ (although the constant in Lemma \ref{bounded_delay} is not tight due to the much better delay we obtain here), which demonstrates the $\mathcal{O}(V)$ delay bound.}

\begin{figure}[htbp]
   \centering
   \includegraphics[height=2in]{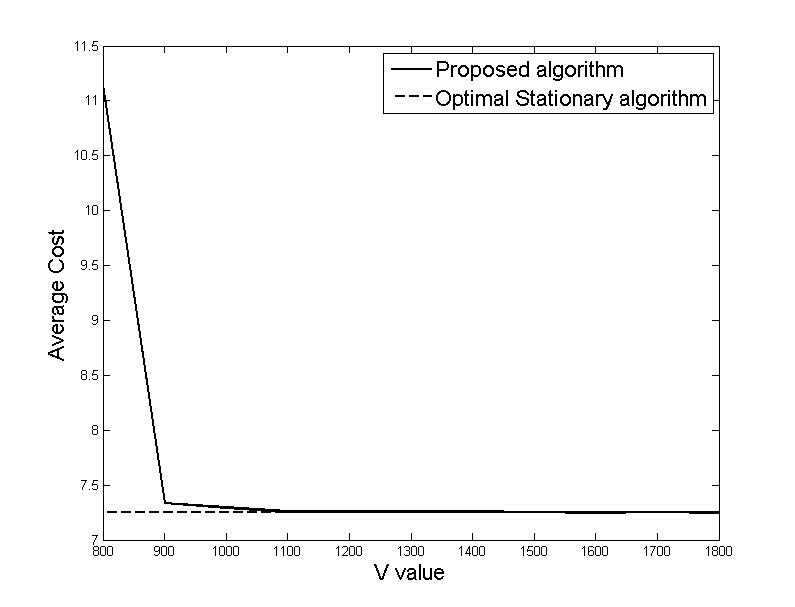} % requires the graphicx package
   \caption{Time average cost verses $V$ parameter over 1 millon slots.}
   \label{fig:Stupendous2}
\end{figure}

\begin{figure}[htbp]
   \centering
   \includegraphics[height=2in]{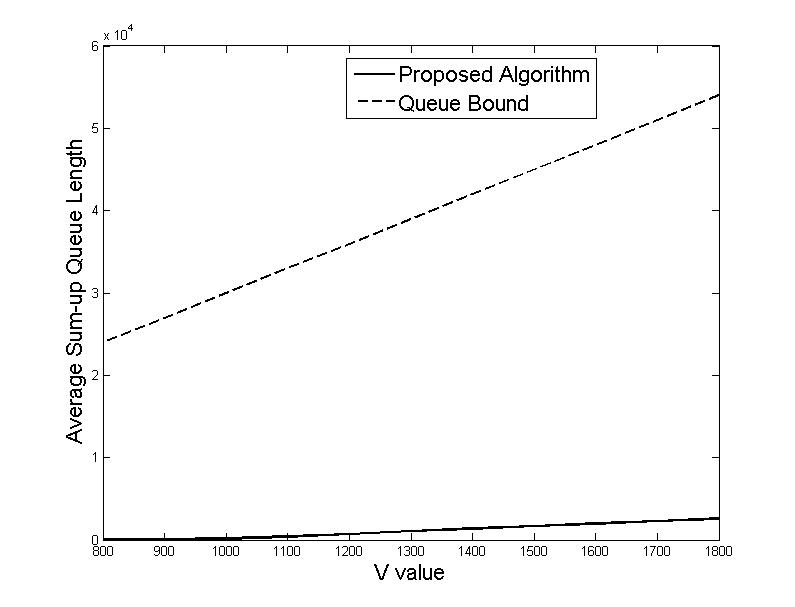} % requires the graphicx package
   \caption{Time average sum-up request queue length verses $V$ parameter over 1 millon slots.}
   \label{fig:Stupendous3}
\end{figure}

We then tune the requests $\lambda(t)$ to be uniform in $[20,40]$ and keep other parameters unchanged. In Fig. \ref{fig:Stupendous4}, we see that since the request rate gets larger, we need $V$ to be larger in order to obtain the near optimality, but still, the near optimality gap scales roughly $\mathcal{O}(1/V)$. Fig. \ref{fig:Stupendous5} gives the sum-up average queue length in this case. The average queue length is larger than that of Fig. \ref{fig:Stupendous3} with linear growth with respect to $V$.

\begin{figure}[htbp]
   \centering
   \includegraphics[height=2in]{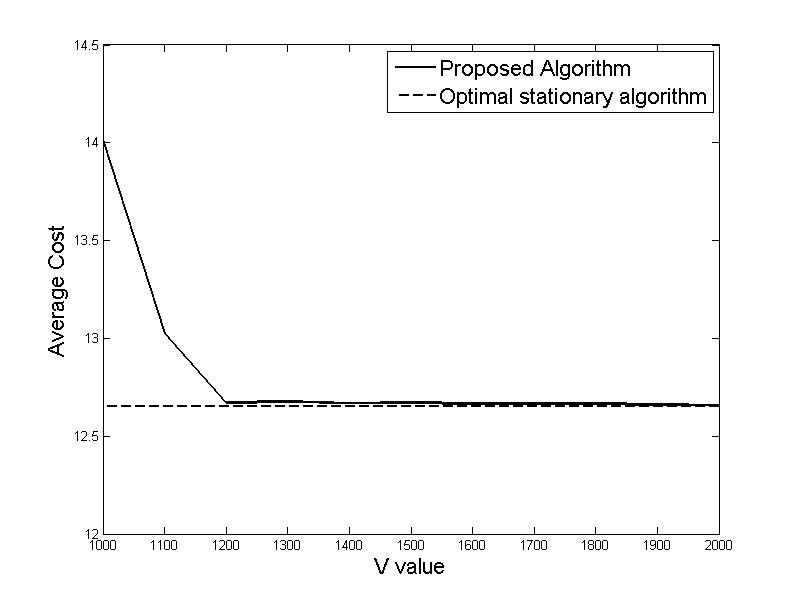} % requires the graphicx package
   \caption{Time average cost verses $V$ parameter over 1 millon slots.}
   \label{fig:Stupendous4}
\end{figure}

\begin{figure}[htbp]
   \centering
   \includegraphics[height=2in]{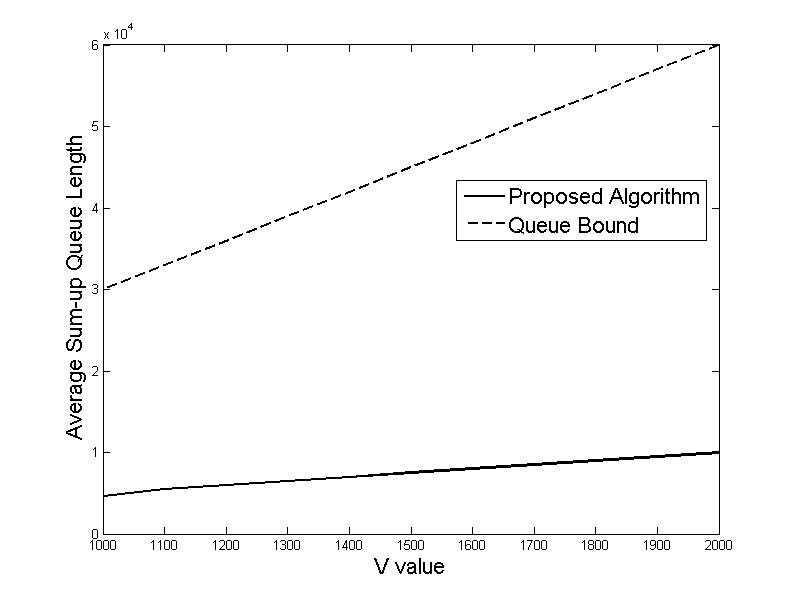} % requires the graphicx package
   \caption{Time average cost verses $V$ parameter over 1 millon slots.}
   \label{fig:Stupendous5}
\end{figure}

\subsection{Real data center traffic trace and performance evaluation}
This second considers a simulation on a real data center traffic obtained from the open source data sets(\cite{real-data-2}) of the paper \cite{real-data-1}. The trace is plotted in Fig. \ref{fig:trace}. We synthesis different data chunks from the source so that the trace contains both the steady phase and increasing phase. The total time duration is 2800 seconds with each slot equal to 20ms. The peak traffic is 2120 requests per 20 ms, and the time average traffic over this whole time interval is 654 requests per 20 ms.

We consider a data center consists of 1060 homogeneous servers. We assume each server has only one sleep state and the service quantity of each server at each slot follows a Zipf's law\footnote{The pdf of Zipf's law with parameter $K,p$ is defined as:
$f(n; K,p)=\frac{1/n^p}{\sum_{i=1}^K1/i^p}$. Thus, the mean of the distribution is $\frac{\sum_{i=1}^K1/i^{p-1}}{\sum_{i=1}^K1/i^p}$.} with parameter $K=10$ and $p=1.9$. This gives the service rate of each server equal to $1.9933\approx2$ requests per 20ms. So the full capacity of the data center is able to support the peak traffic.
Zipf's law is previous introduced to model a wide scope of physics, biology, computer science and social science phenomenon(\cite{Zipf-law}), and is adopted by various literatures to simulate the empirical data center service rate(\cite{Gandhi_phd_thesis, sleep-mode-2}). 
The setup time of each server is geometrically distributed with success probability equal to $0.001$. This gives the mean setup time 1000 slots (20 seconds). This setup time is previously shown in \cite{sleep-mode-2} to be a typical time duration for a desktop to recover from the suspend or hibernate state. 

Furthermore, to make a fair comparison with several existing algorithms, we enforce the front end balancer to accept all requests at each time slot (so the rejection rate is always 0). The only cost in the system is then the power consumption. We assume that a server consumes 10 W each slot when active and 0 W each slot when idle. The setup cost is also 10 W per slot. Moreover, we apply the one queue model described in Section \ref{section_improve} for all the rest of the simulations. Following the problem formulation, the maximum idle time of a server for the proposed algorithm is $I_{\max}=5000$, while no such limit is imposed for any other benchmark algorithms.

We first run our proposed algorithm over the trace with virtualization (in Section \ref{section_improve}) for different $V$ values. We set the initial virtual queue backlog $Q_n(0)=2000~\forall n$, and keep 20 servers always on.
 Fig. \ref{fig:costV} and Fig. \ref{fig:queueV} plots the running average power consumption and corresponding queue length for $V=400,~600,~800$ and 1200, respectively. It can be seen that as $V$ gets large, the average power consumption does not improve too much but the queue length changes drastically. This phenomenon results from the $[\mathcal{O}(1/V), \mathcal{O}(V)]$ trade-off of our proposed algorithm. In view of this fact, we choose $V=600$ which gives a reasonable delay performance in Fig. \ref{fig:queueV}.

\begin{figure}[htbp]
   \centering
   \includegraphics[height=2in]{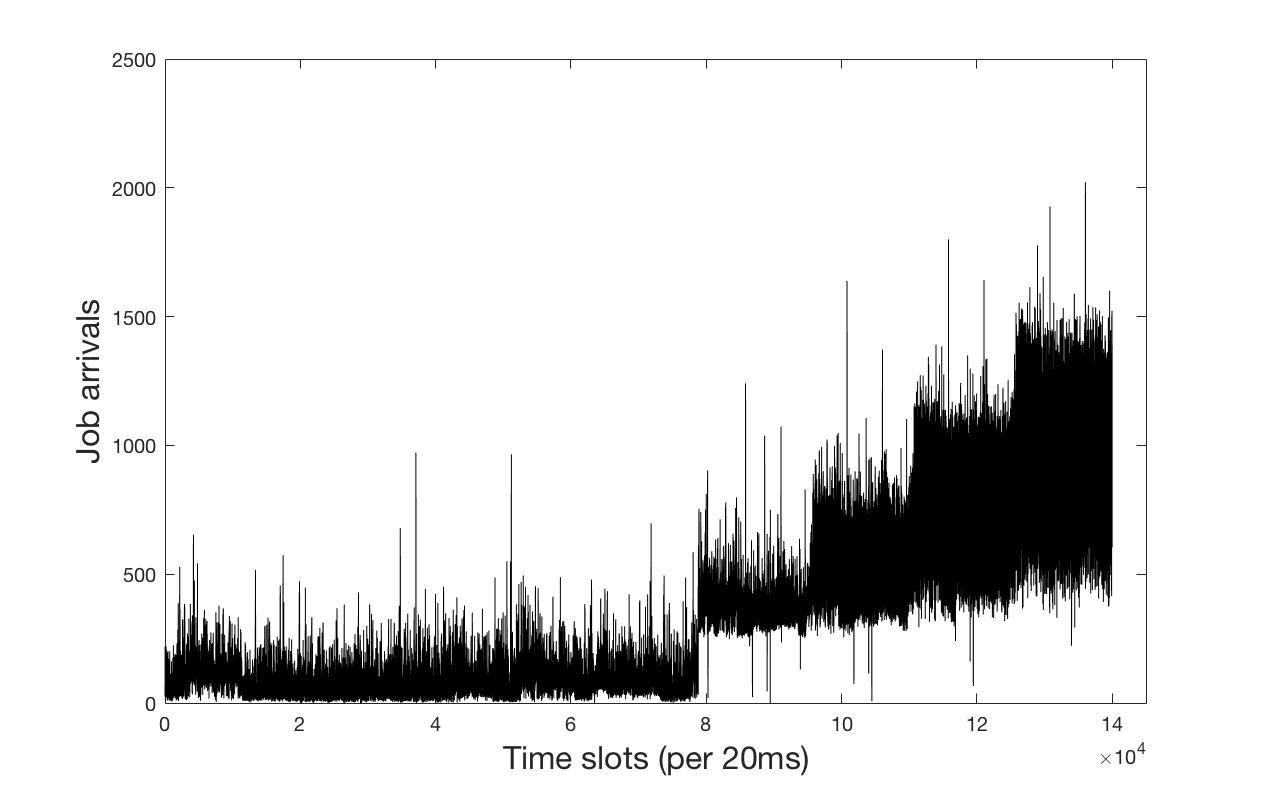} % requires the graphicx package
   \caption{Synthesized traffic trace from \cite{real-data-2}.}
   \label{fig:trace}
\end{figure}

\begin{figure}[htbp]
   \centering
   \includegraphics[height=2in]{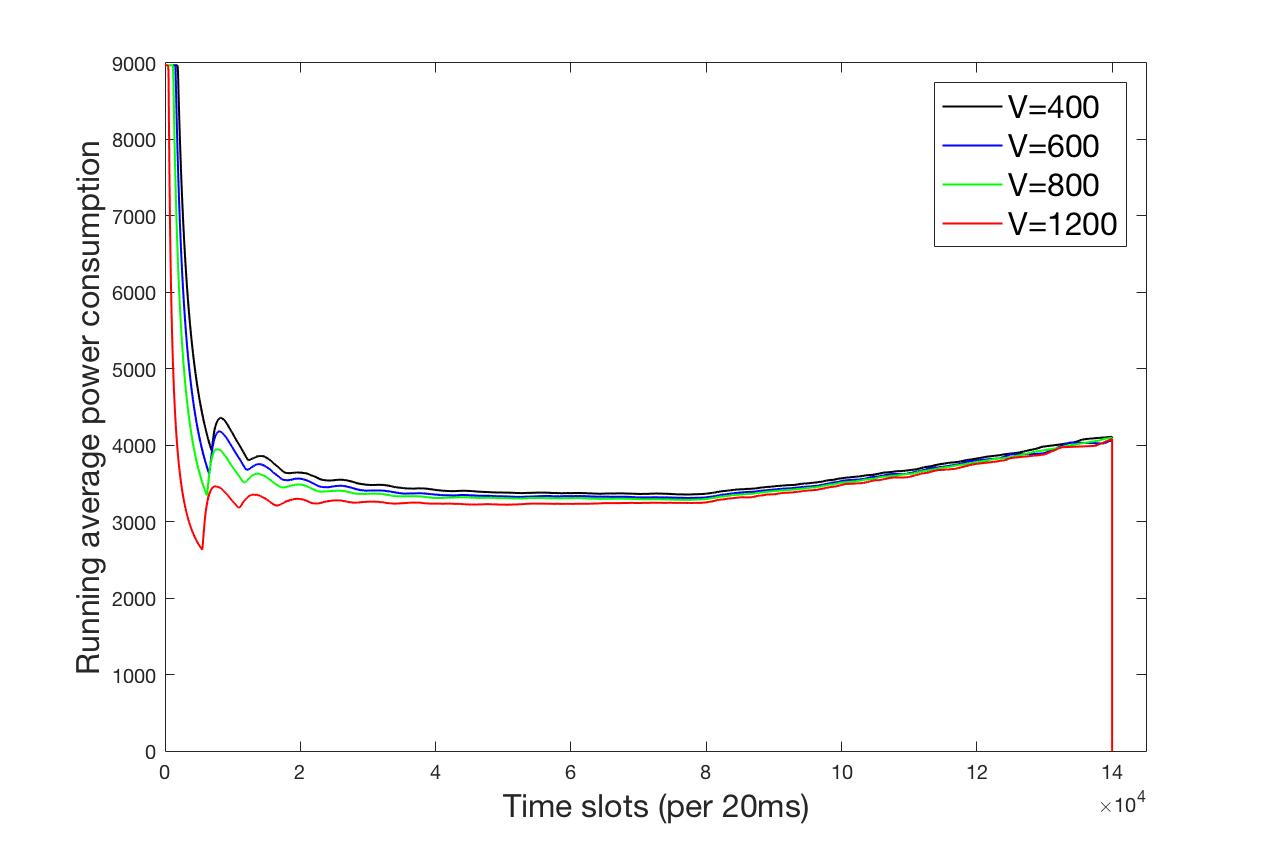} % requires the graphicx package
   \caption{Running average power consumption from slot 1 to the current slot for different $V$ value.}
   \label{fig:costV}
\end{figure}

\begin{figure}[htbp]
   \centering
   \includegraphics[height=2in]{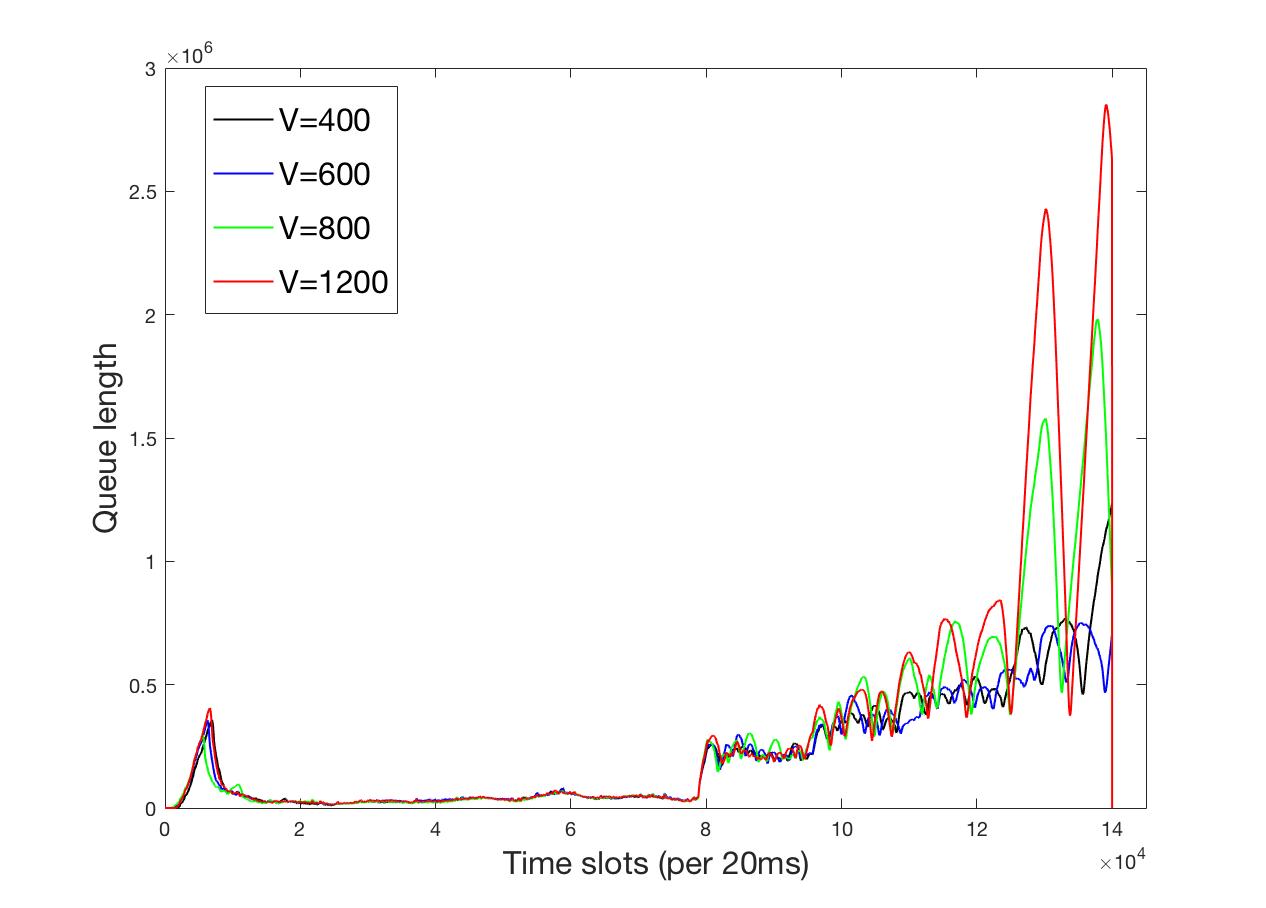} % requires the graphicx package
   \caption{Instantaneous queue length for different $V$ value.}
   \label{fig:queueV}
\end{figure}

Next, we compare our proposed algorithm with the same initial setup and $V=600$ to the following algorithms:
\begin{itemize}
\item Always-on with $N=327$ active servers and the rest servers staying on the sleep mode. Note that 327 servers can support the average traffic over the whole interval which is 654 requests per 20 ms. 
\item Always-on with full capacity. This corresponds to keeping all 1060 servers on at every slot.
\item Reactive. This algorithm is developed in \cite{sleep-mode-2} which reacts to the current traffic     $\overline{\lambda}(t)$ and maintains $k_{react}(t)=\left\lceil\overline{\lambda}(t)/2\right\rceil$ servers on. In the simulation, we choose $\overline{\lambda}(t)$ to be the average of the traffic from the latest 10 slots. If the current active server $k(t)>k_{react}(t)$, then, we turn $k(t)-k_{react}(t)$ servers off, otherwise, we turn $k_{react}(t)-k(t)$ servers to the setup state.
\item Reactive with extra capacity. This algorithm is similar to Reactive except that we introduce a virtual traffic flow of $p$ jobs per slot. So during each time slot $t$, the algorithm maintains $k_{react}(t)=\left\lceil(\overline{\lambda}(t)+p)/2\right\rceil$ servers on. 
\end{itemize}
Fig. \ref{fig:Stupendous7}-\ref{fig:Stupendous9} plots the average power consumption, queue length and the number of active servers, respectively. It can be seen that all algorithms perform pretty well during first half of the trace. For the second half of the trace, the traffic load is increasing. The Always-on algorithm with mean capacity does not adapt to the traffic so the queue length blows up quickly.
Because of the long setup time, the number of active servers in the Reactive algorithm fails to catch up with the increasing traffic so the queue length also blows up. Our proposed algorithm minimizes the power consumption while stabilizing the queues, thereby outperforms both the Always-on and the Reactive algorithm. Note that the Reactive with extra 200 job capacity is able to achieve a similar delay performance as our proposed algorithm, but with significant extra power consumption.

\begin{figure}[htbp]
   \centering
   \includegraphics[height=2.3in]{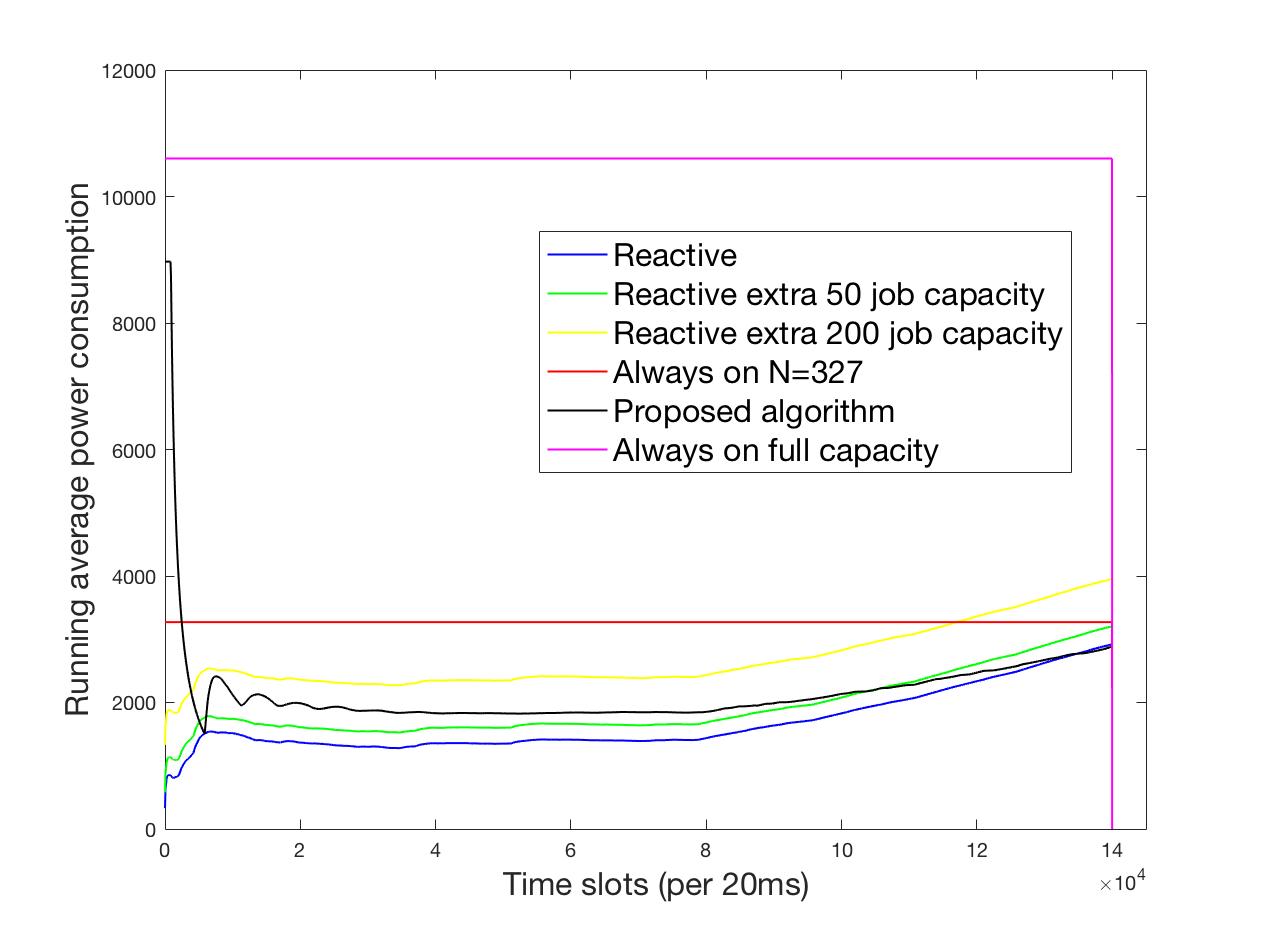} % requires the graphicx package
   \caption{Running average power consumption from slot 1 to the current slot for different algorithms.}
   \label{fig:Stupendous7}
\end{figure}

\begin{figure}[htbp]
   \centering
   \includegraphics[height=2in]{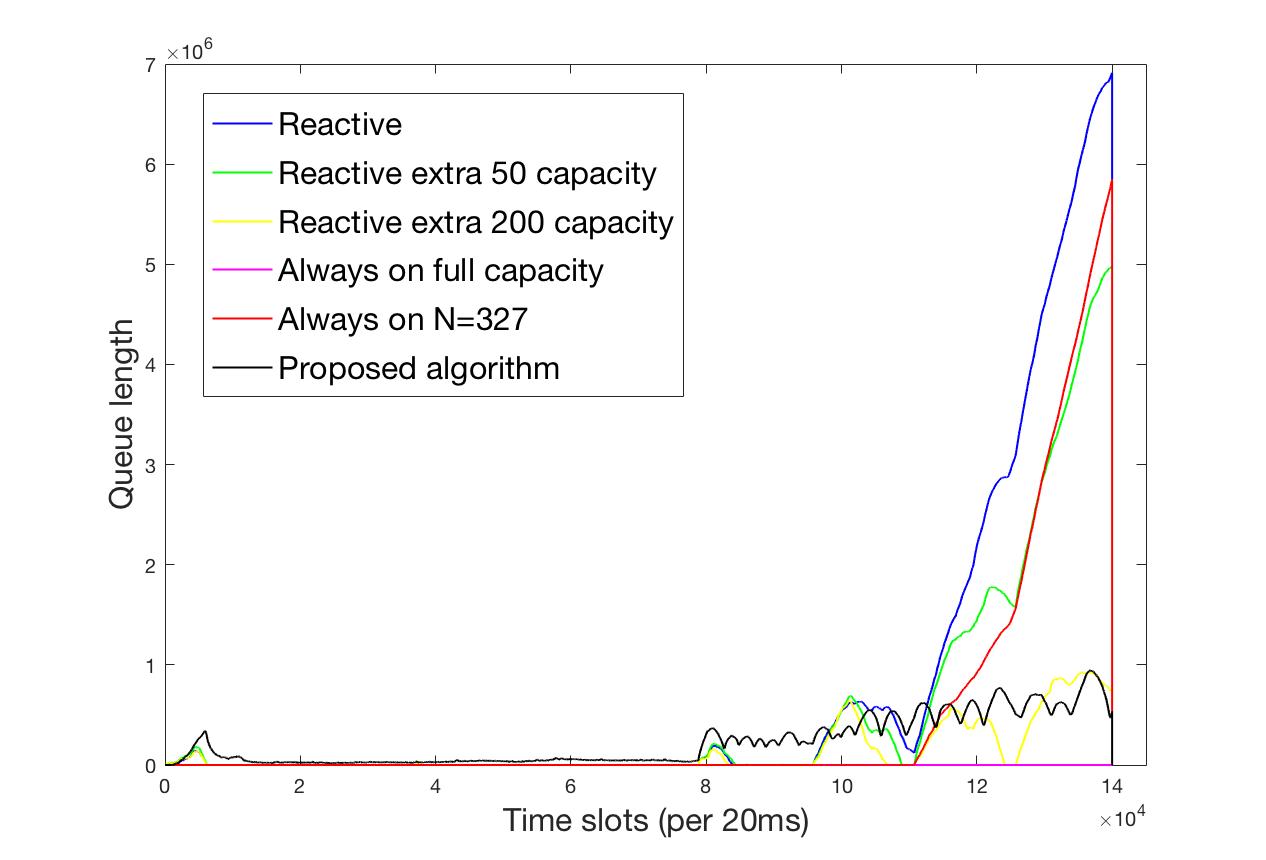} % requires the graphicx package
   \caption{Instantaneous queue length for different algorithms.}
   \label{fig:Stupendous8}
\end{figure}

\begin{figure}[htbp]
   \centering
   \includegraphics[height=2in]{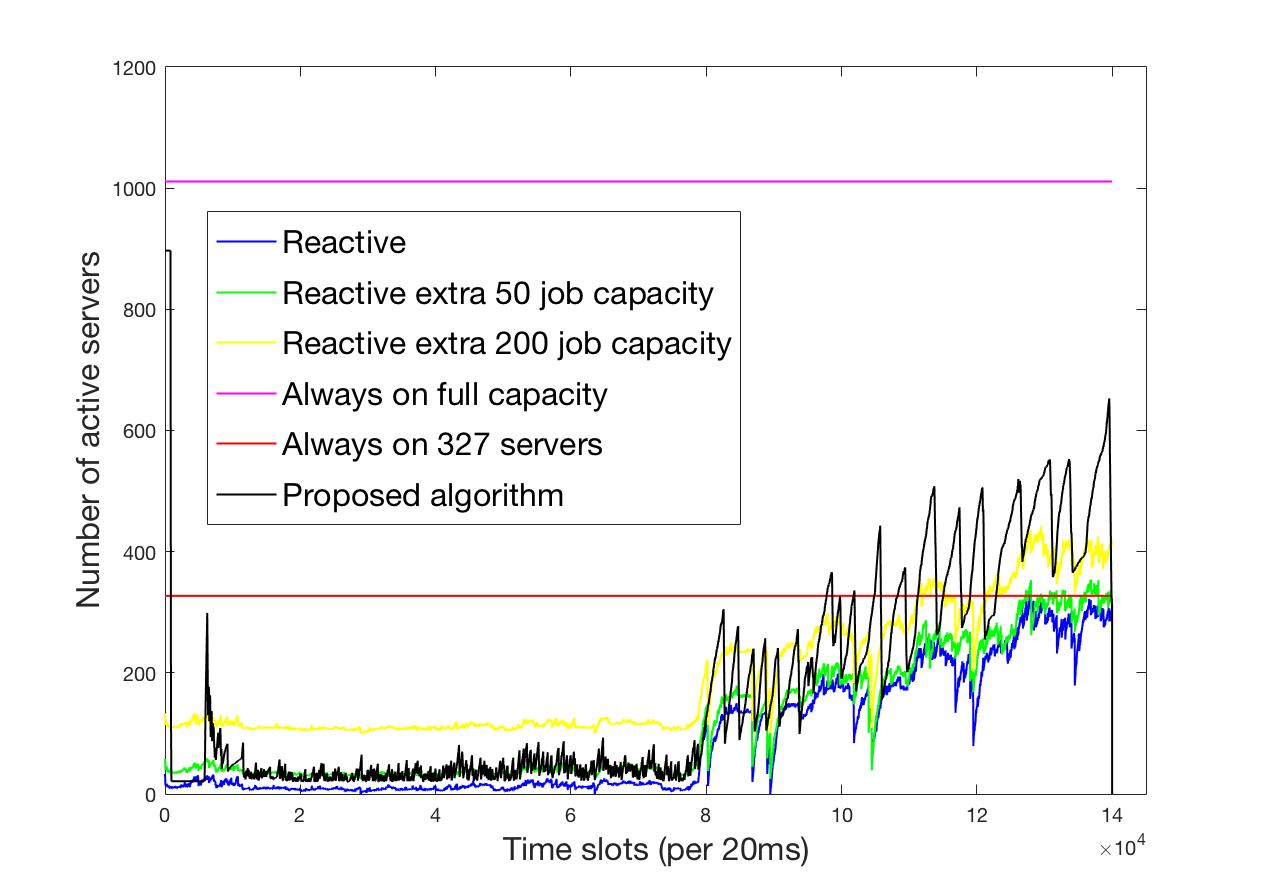} % requires the graphicx package
   \caption{Number of active servers over time.}
   \label{fig:Stupendous9}
\end{figure}

Finally, we evaluate the influence of different sleep modes on the performance. We keep all the setups the same as before and consider the sleep modes with sleep power consumption  equal to 2 W and 4 W per slot, respectively. Since the Always-on and the Reactive algorithm do not look at the sleep power consumption, their decisions remain the same as before, thus, we superpose the queue length of our proposed algorithm onto the previous Fig. \ref{fig:Stupendous8} and get the queue length comparison in Fig. \ref{fig:queue-length-S}. We see from the plot that increasing the power consumption during the sleep mode only slightly increases the queue length of our proposed algorithm. Fig. \ref{fig:sleep-mode} plots the running average power consumption under different sleep modes. Despite spending more power on the sleep mode, the proposed algorithm can still save considerable amount of power compared to other algorithms while keeping the request queue stable. This shows that our algorithm is empirically robust to the change of sleep mode.
\begin{figure}[htbp]
   \centering
   \includegraphics[height=2in]{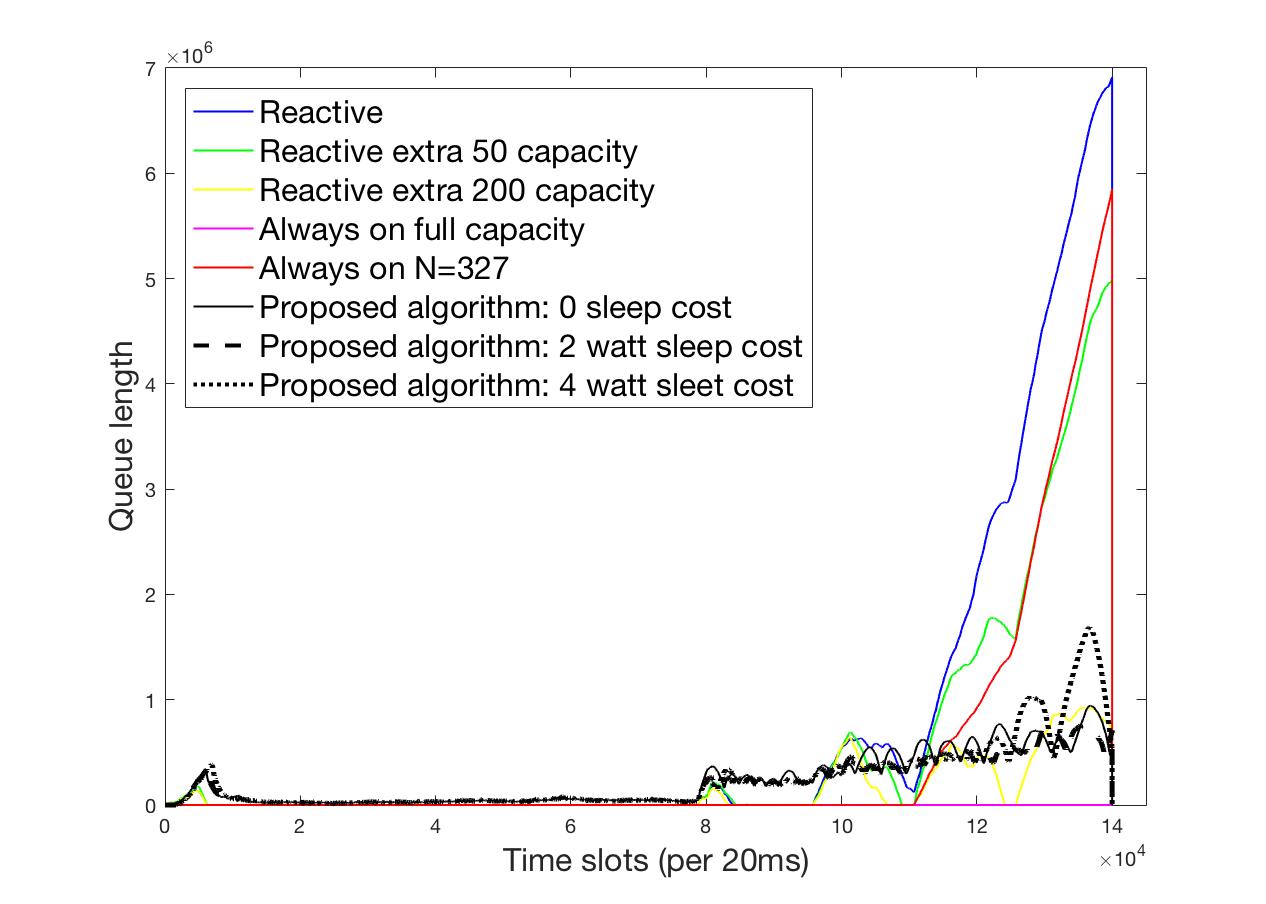} % requires the graphicx package
   \caption{Instantaneous queue length for different algorithms.}
   \label{fig:queue-length-S}
\end{figure}

\setcounter{figure}{15}
\begin{figure*}[ht]
\centering
 \begin{minipage}{5.6cm}
   \includegraphics[height=4cm] {average_cost}
 \end{minipage}
 \begin{minipage}{5.6cm}
   \includegraphics[height=4cm] {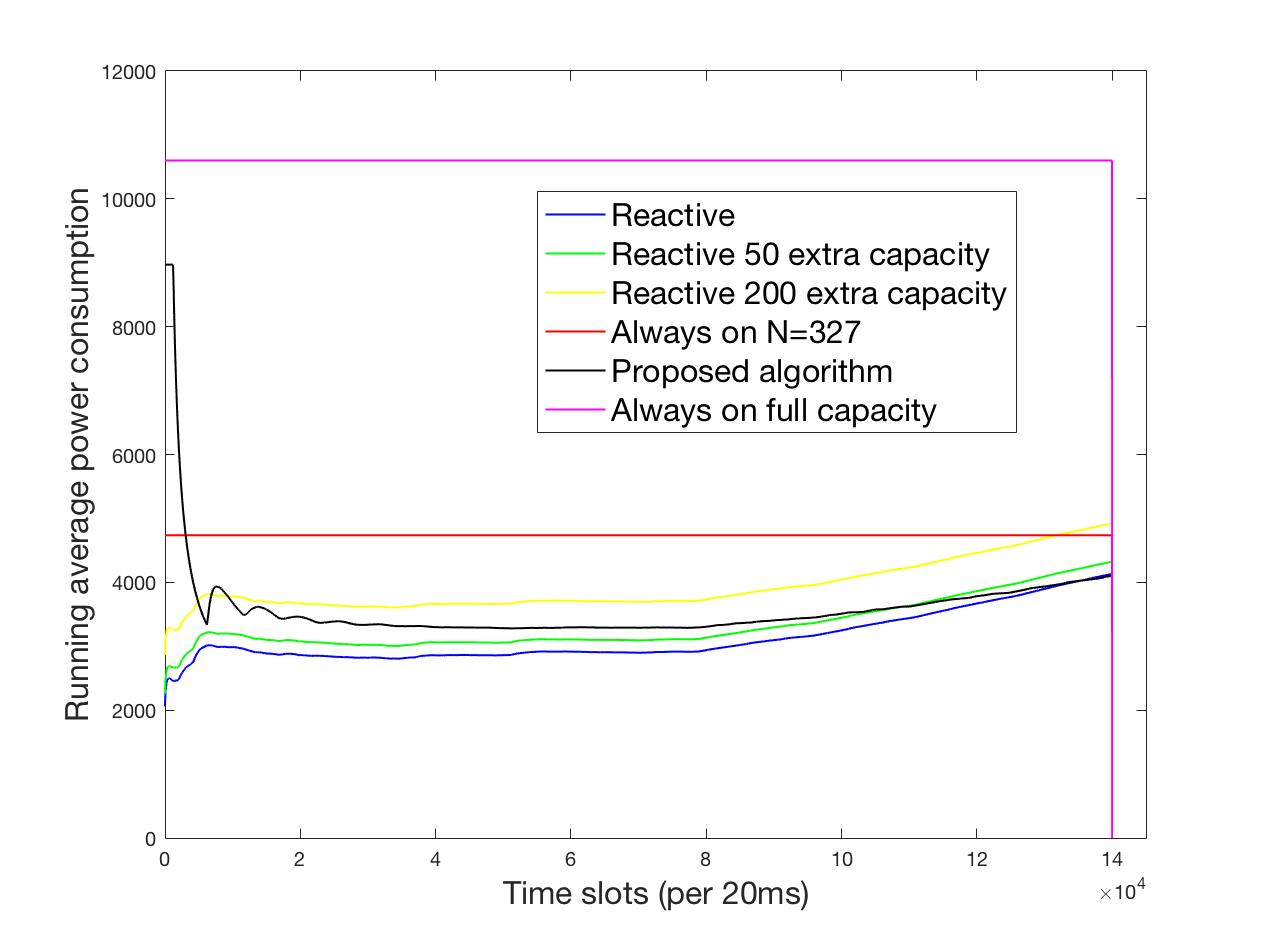}
 \end{minipage}
 \begin{minipage}{5.6cm}
   \includegraphics[height=4cm] {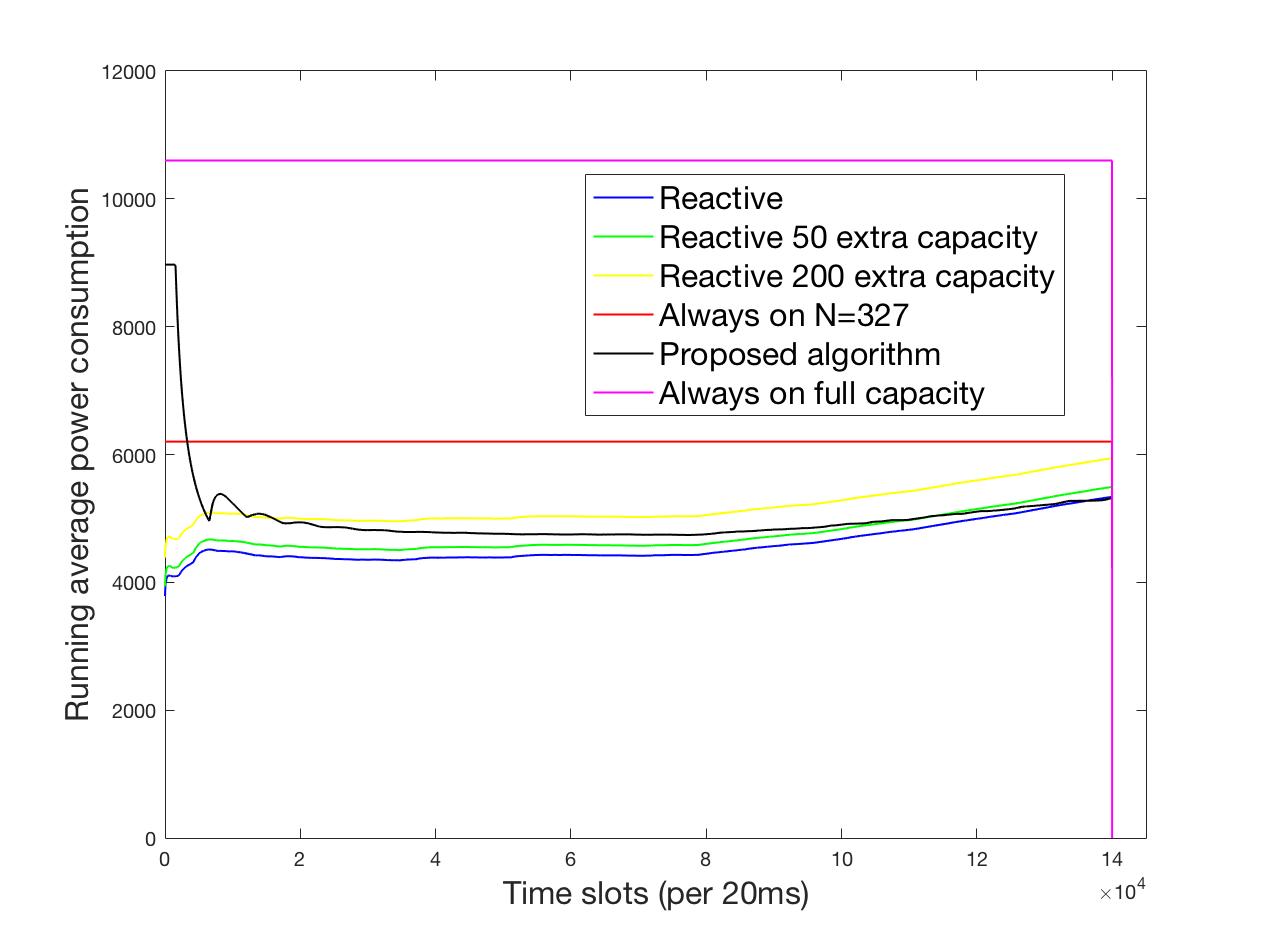}
 \end{minipage}
\caption{Running average power consumption for 0 W sleep cost(left), 2 W sleep cost(middle) and 4 W sleep cost(right)}\label{fig:sleep-mode}
\end{figure*}

\section{Conclusions}
This paper proposes an efficient distributed asynchronous control algorithm reducing the cost in a data center, where the front-end load balancer makes slot-wise routing requests to the shortest queue and each server makes frame-based service decision by only looking at its own request queue. Theoretically, this algorithm is shown to achieve the near optimal cost while stabilizing the request queues. Simulation experiments on a real data center traffic trace demonstrates that our algorithm outperforms several other algorithms in reducing the power consumption as well as achieving lower delays.

\section*{Appendix A--- Proof of Lemma \ref{compute_idle}}
We first show that the optimal decision on $I_n[f]$ is a pure decision. We have \eqref{appendix_A_interim} holds,
\begin{figure*}
\normalsize
\begin{align}
&D_n[f]=\frac{V\hat{W}_n(\alpha_n[f])m_{\alpha_n[f]}+Ve_n-Q_n(t_f^{(n)})\mu_n
        +\expect{\frac{B_0}{2}(I_n[f]+\tau_n[f]+1)^2+V\hat{g}(\alpha_n[f])I_n[f]\left|~Q_n(t_f^{(n)})\right.}}
        {\expect{I_n[f]+\tau_n[f]+1~\left|~Q_n(t_f^{(n)})\right.}}-\frac{B_0}{2}\nonumber\\
&=\frac{V\hat{W}_n(\alpha_n[f])m_{\alpha_n[f]}+Ve_n-Q_n(t_f^{(n)})\mu_n+\expect{\frac{B_0}{2}(I_n[f]+m_{\alpha_n}+1)^2
  +\frac{B_0}{2}\sigma_{\alpha_n[f]}^2+V\hat{g}(\alpha_n[f])I_n[f]~\left|~Q_n(t_f^{(n)})\right.}}
  {\expect{I_n[f]+m_{\alpha_n}+1~\left|~Q_n(t_f^{(n)})\right.}}-\frac{B_0}{2}\label{appendix_A_interim}
\end{align}
\end{figure*}
where the first equality follows from the definition $T_n[f]=I_n[f]+\tau_n[f]+1$ and the second equality follows from the fact that the setup time $\tau_n[f]=\hat{\tau}_n(\alpha_n[f])$ is independent of $I_n[f]$ with mean $m_{\alpha_n[f]}$ and variance $\sigma_{\alpha_n[f]}^2$.
For simplicity of notations, let
\begin{align*}
F(\alpha_n[f],I_n[f]) =& V\hat{W}_n(\alpha_n[f])m_{\alpha_n[f]}+Ve_n-Q_n(t_f^{(n)})\mu_n\\
        &+\frac{B_0}{2}(I_n[f]+m_{\alpha_n[f]}+1)^2+V\hat{g}(\alpha_n[f])I_n[f]\\
        &+\frac{B_0}{2}\sigma_{\alpha_n[f]}^2\\
G(\alpha_n[f],I_n[f]) =& I_n[f]+m_{\alpha_n[f]}+1,
\end{align*}
then
\[D_n[f]=\frac{\expect{F(\alpha_n[f],I_n[f])~|~Q_n(t_f^{(n)})}}{\expect{G(\alpha_n[f],I_n[f])~|~Q_n(t_f^{(n)})}}-\frac{B_0}{2}.\]
Meanwhile, given the queue length $Q_n(t_f^{(n)})$ at frame $f$, denote the benchmark solution over pure decisions as
\begin{equation}\label{det_solution}
m \triangleq \min_{I_n[f]\in\mathbb{N},~I_n[f]\in[1,I_{\max}],\alpha_n[f]\in\mathcal{L}_n}\frac{F(\alpha_n[f],I_n[f])}{G(\alpha_n[f],I_n[f])}.
\end{equation}
Then, for any randomized decision on $\alpha_n[f]$ and $I_n[f]$, its realization within frame $f$ satisfies the following
\[\frac{F(\alpha_n[f],I_n[f])}{G(\alpha_n[f],I_n[f])}\geq m,\]
which implies
\[F(\alpha_n[f],I_n[f])\geq m G(\alpha_n[f],I_n[f]).\]
Taking conditional expectation from both sides gives
\begin{align*}
&\expect{F(\alpha_n[f],I_n[f])~|~Q_n(t_f^{(n)})}\\
&\geq m\expect{G(\alpha_n[f],I_n[f])~|~Q_n(t_f^{(n)})}\\
&\Rightarrow~\frac{\expect{F(\alpha_n[f],I_n[f])~|~Q_n(t_f^{(n)})}}{\expect{G(\alpha_n[f],I_n[f])~|~Q_n(t_f^{(n)})}}\geq m.
\end{align*}
Thus, it is enough to consider pure decisions only, which boils down to computing \eqref{det_solution}. This proves the lemma.

\section*{Appendix B--- Proof of Lemma \ref{bounded_supMG}}
This section is dedicated to prove that $\expect{X_n[f]^2}$ is bounded. First of all, since the idle option set $\mathcal{L}_n$ is finite, denote
\begin{align*}
W_{\max}=\max_{\alpha_n\in\mathcal{L}_n}W_n(\alpha_n)\\
g_{\max}=\max_{\alpha_n\in\mathcal{L}_n}g_n(\alpha_n)
\end{align*}
It is obvious that $|W_n(t)-\overline{W}_n^*|\leq W_{\max}$, $|g_n(t)-\overline{G}_n^*|\leq g_{\max}$, $|e_nH_n(t)-\overline{E}_n^*|\leq e_n$, and $|\mu_nH_n(t)-\overline{\mu}^*|\leq\mu_n$. Combining with the boundedness of queues in lemma \ref{bounded_delay}, it follows
\begin{align*}
|X_n[f]|\leq&\sum_{t=t_f^{(n)}}^{t=t_{f+1}^{(n)}-1}\left(V\left(W_{\max}+e_n+g_{\max}\right)
       +\left(Vc_{\max}\right.\right.\\
       &\left.\left.+R_{\max}\right)\mu_n+\left(t-t_f^{(n)}\right)B_0+\Psi_n\right)\\
      \leq&\left(V(W_{\max}+e_n+g_{\max})+(Vc_{\max}+R_{\max})\mu_n\right.\\
       &\left.+\Psi_n\right)T_n[f]+\frac{T_n[f](T_n[f]-1)B_0}{2}
\end{align*}
Let $B_1\triangleq V(W_{\max}+e_n+g_{\max})+(Vc_{\max}+R_{\max})\mu_n+\Psi_n+B_0/2$, it follows
\[|X_n[f]|\leq B_1T_n[f]+\frac{B_0}{2}T_n[f]^2.\]
Thus,
\[\expect{X_n[f]^2}\leq B_1^2\expect{T_n[f]^2}+B_1B_0\expect{T_n[f]^3}+\frac{B_0^2}{4}\expect{T_n[f]^4}.\]
Notice that $T_n[f]\leq I_n[f]+\tau_n[f]+1$ by \eqref{frame_length}, where $I_n[f]$ is upper bonded by $I_{\max}$ and $\tau_n[f]$ has first four moments bounded by assumption \ref{bounded_moment_assumption}. Thus, $\expect{X_n[f]^2}$ is bounded by a fixed constant.

\section*{Appendix C--- Proof of Lemma \ref{true_time_average}}
\begin{proof}
Let's first abbreviate the notation by defining
\begin{align*}
Y(t)=&V(W_n(t)+e_nH_n(t)+g_n(t))-Q_n(t_f^{(n)})
(\mu_nH_n(t)-\overline{\mu}_n^*)\\
&+\left(t-t_f^{(n)}\right)B_0.
\end{align*}
For any $T\in[t_F^{(n)},~t_F^{(n+1)})$, we can bound the partial sums from above by the following
\[\sum_{t=0}^{T-1}Y(t)\leq\sum_{t=0}^{t_F^{(n)}-1}Y(t)+B_2T_n[F]+\frac{B_0}{2}T_n[F]^2,\]
where $B_0=\frac{1}{2}(R_{\max}+\mu_{\max})\mu_{\max}$ is defined in \eqref{server_decision}, and $B_2\triangleq VW_n+V\mu_ne_n+(Vc_{\max}+R_{\max})\mu_n+B_0/2$. Thus,
\begin{align*}
\frac{1}{T}\sum_{t=0}^{T-1}Y(t)&\leq\frac{1}{T}\sum_{t=0}^{t_F^{(n)}-1}Y(t)+\frac{1}{T}\left(B_2T_n[F]+\frac{B_0}{2}T_n[F]^2\right)\\
&\leq\max\{a[F],~b[F]\},
\end{align*}
where
\begin{align*}
a[F]\triangleq&\frac{1}{t_F^{(n)}}\sum_{t=0}^{t_F^{(n)}-1}Y(t)+\frac{1}{t_F^{(n)}}\left(B_2T_n[F]+\frac{B_0}{2}T_n[F]^2\right),\\
b[F]\triangleq&\frac{1}{t_{F+1}^{(n)}}\sum_{t=0}^{t_F^{(n)}-1}Y(t)+\frac{1}{t_{F+1}^{(n)}}\left(B_2T_n[F]+\frac{B_0}{2}T_n[F]^2\right).
\end{align*}
Thus, this implies that
\begin{align*}
\limsup_{T\rightarrow\infty}\frac{1}{T}\sum_{t=0}^{T-1}Y(t)
\leq&\limsup_{F\rightarrow\infty}\max\{a[F],~b[F]\}\\
=&\max\left\{\limsup_{F\rightarrow\infty}a[F],~\limsup_{F\rightarrow\infty}b[F]\right\}.
\end{align*}
We then try to work out an upper bound for $\limsup_{F\rightarrow\infty}a[F]$ and $\limsup_{F\rightarrow\infty}b[F]$ respectively.

\begin{enumerate}
  \item Bound for $\limsup_{F\rightarrow\infty}a[F]$:
\begin{align*}
\limsup_{F\rightarrow\infty}a[F]\leq&\limsup_{F\rightarrow\infty}\frac{1}{t_F^{(n)}}\sum_{t=0}^{t_F^{(n)}-1}Y(t)\\
                                    &+\limsup_{F\rightarrow\infty}\frac{1}{t_F^{(n)}}\left(B_2T_n[F]+\frac{B_0}{2}T_n[F]^2\right)\\
                                \leq& V(\overline{W}_n^*+\overline{E}_n^*+\overline{G}_n^*)+\Psi_n\\
                                    &+\limsup_{F\rightarrow\infty}\frac{1}{t_F^{(n)}}\left(B_2T_n[F]+\frac{B_0}{2}T_n[F]^2\right).
\end{align*}
where the second inequality follows from corollary \ref{corollary_ratio_time_average}.
It remains to show that
\begin{equation}\label{interim_a[F]}
\limsup_{F\rightarrow\infty}\frac{1}{t_F^{(n)}}\left(B_2T_n[F]+\frac{B_0}{2}T_n[F]^2\right)\leq0.
\end{equation}
Since $t_F^{(n)}\geq F$, it is enough to show that
\begin{align}
&\limsup_{F\rightarrow\infty}\frac{T_n[F]}{F}=0,   \label{a_1}\\
&\limsup_{F\rightarrow\infty}\frac{T_n[F]^2}{F}=0. \label{a_2}
\end{align}
We prove \eqref{a_2}, and \eqref{a_1} is similar. Since each $T_n[F]=I_n[F]+\tau_n[F]+1$, where $I_n[F]\leq I_{\max}$ and $\tau_n[F]$ has bounded first four moments, the first four moments of $T_n[F]$ must also be bounded and there exists a constant $C>0$ such that
\[\expect{T_n[F]^4}\leq C.\]
For any $\epsilon>0$, define a sequence of events
\[A_F^\epsilon\triangleq\left\{T_n[F]^2>\epsilon F\right\}.\]
According to Markov inequality,
\[Pr\left[A_F^\epsilon\right]\leq\frac{\expect{T_n[F]^4}}{\epsilon^2F^2}\leq\frac{C}{\epsilon^2F^2}.\]
Thus,
\[\sum_{F=1}^\infty Pr\left[A_F^\epsilon\right]\leq\frac{C}{\epsilon^2}\sum_{F=1}^\infty\frac{1}{F^2}\leq\frac{2C}{\epsilon^2}<\infty. \]
By Borel-Cantelli lemma (lemma 1.6.1 in \cite{durrett_probability}),
\[Pr\left[A_F^\epsilon~\textrm{occurs infinitely often}\right]=0,\]
which implies
\[Pr\left[\limsup_{F\rightarrow\infty}\frac{T_n[F]^2}{F}>\epsilon\right]=0.\]
Since $\epsilon$ is arbitrary, this implies \eqref{a_2}. Similarly, \eqref{a_1} can be proved. Thus, \eqref{interim_a[F]} holds and
\[\limsup_{F\rightarrow\infty}a[F]\leq V(\overline{W}_n^*+\overline{E}_n^*+\overline{G}_n^*)+\Psi_n.\]

\item Bound for $\limsup_{F\rightarrow\infty}b[F]$:
\begin{align*}
\limsup_{F\rightarrow\infty}b[F]\leq&\limsup_{F\rightarrow\infty}\frac{1}{t_F^{(n)}}\sum_{t=0}^{t_F^{(n)}-1}Y(t)\cdot\frac{t_F^{(n)}}{t_{F+1}^{(n)}}\\
                                    &+\limsup_{F\rightarrow\infty}\frac{1}{t_{F+1}^{(n)}}\left(B_2T_n[F]+\frac{B_0}{2}T_n[F]^2\right).\\
                                \leq&\limsup_{F\rightarrow\infty}\left(\frac{1}{t_F^{(n)}}\sum_{t=0}^{t_F^{(n)}-1}Y(t)\right)\\
                                    &\cdot\limsup_{F\rightarrow\infty}\frac{t_F^{(n)}}{t_{F+1}^{(n)}}\\
                                \leq&\left(V\left(\overline{W}_n^*+\overline{E}_n^*+\overline{G}_n^*\right)+\Psi_n\right)\\
                                    &\cdot\limsup_{F\rightarrow\infty}\frac{t_F^{(n)}}{t_{F+1}^{(n)}}\\
                                \leq&V\left(\overline{W}_n^*+\overline{E}_n^*+\overline{G}_n^*\right)+\Psi_n,
\end{align*}
where the second inequality follows from \eqref{interim_a[F]}, the third inequality follows from corollary \ref{corollary_ratio_time_average} and the last inequality follows from the fact that $V\left(\overline{W}_n^*+\overline{E}_n^*+\overline{G}_n^*\right)+\Psi_n>0$.
\end{enumerate}
Above all, we proved the lemma.
\end{proof}

\section*{Appendix D--- Proof of Theorem \ref{theorem_near_optimal_perform}}
\begin{proof}
Define the drift-plus-penalty(DPP) expression $P(t)$ as follows
\begin{align*}
P(t)=&V\left(c(t)d(t)+\sum_{n=1}^N\left(W_n(t)+e_nH_n(t)+g_n(t)\right)\right)\\
&+\frac{1}{2}\sum_{n=1}^N\left(Q_n(t+1)^2-Q_n(t)^2\right).
\end{align*}
By simple algebra using the queue updating rule \eqref{queue_update}, we can work out the upper bound for $P(t)$ as follows,
\begin{align*}
P(t)\leq& \frac{1}{2}\sum_{n=1}^N(R_n(t)+\mu_n)^2+V\left(c(t)d(t)+\right.\\
        &  \left.\sum_{n=1}^N\left(W_n(t)+e_nH_n(t)+g_n(t)\right)\right)\\
        &  +\sum_{n=1}^NQ_n(t)(R_n(t)-\mu_nH_n(t))\\
    \leq& B_3+V\left(c(t)d(t)+\sum_{n=1}^N\left(W_n(t)+e_nH_n(t)+g_n(t)\right)\right)\\
        & +\sum_{n=1}^NQ_n(t)(R_n(t)-\mu_nH_n(t))\\
\end{align*}
\begin{align*}
    \leq& B_3+Vc(t)d(t)+\sum_{n=1}^NQ_n(t)\left(R_n(t)-\overline{R}^*_n\right)\\
          &+V\sum_{n=1}^N\left(W_n(t)+e_nH_n(t)+g_n(t)\right)\\
          &+\sum_{n=1}^NQ_n(t)\left(\overline{\mu}^*_n-\mu_nH_n(t)\right)
\end{align*}
where $B_3=\frac{1}{2}\sum_{n=1}^N(R_{\max}+\mu_n)^2$, the last inequality follows from adding $\sum_{n=1}^NQ_n(t)\overline{\mu}^*_n$ and subtracting $\sum_{n=1}^NQ_n(t)\overline{R}^*$ with the fact that the best randomized stationary algorithm should also satisfy the constraint \eqref{obj_4}, i.e. $\overline{\mu}^*\geq\overline{R}_n^*$.

Now we take the partial average of $P(t)$ from 0 to $T-1$ and take $\limsup_{T\rightarrow\infty}$,
\begin{align}
&\limsup_{T\rightarrow\infty}\frac{1}{T}\sum_{n=1}^{T-1}P(t)\nonumber \\
\leq& B_3 +
\limsup_{T\rightarrow\infty}\frac{1}{T}\sum_{t=0}^{T-1}\left(Vc(t)d(t)+\sum_{n=1}^NQ_n(t)\left(R_n(t)-\overline{R}^*_n\right)\right)\nonumber\\
&+\sum_{n=1}^N\limsup_{T\rightarrow\infty}\frac{1}{T}\sum_{t=0}^{T-1}\left(V\left(W_n(t)+e_nH_n(t)+g_n(t)\right)\right.\nonumber\\
&\left.+Q_n(t)\left(\overline{\mu}^*_n-\mu_nH_n(t)\right)\right).
\label{DPP_bound}
\end{align}
According to \eqref{prob_1_front_end},
\begin{equation}\label{DPP_sub_bound_1}
\limsup_{T\rightarrow\infty}\frac{1}{T}\sum_{t=0}^{T-1}\left(Vc(t)d(t)+\sum_{n=1}^NQ_n(t)\left(R_n(t)-\overline{R}^*_n\right)\right)
\leq V\overline{C}^*.
\end{equation}
On the other hand,
\begin{align}
&\limsup_{T\rightarrow\infty}\frac{1}{T}\sum_{t=0}^{T-1}\left(V\left(W_n(t)+e_nH_n(t)+g_n(t)\right)\right.\nonumber\\
      &\left.+Q_n(t)\left(\overline{\mu}^*_n-\mu_nH_n(t)\right)\right)\nonumber\\
\leq&\limsup_{T\rightarrow\infty}\frac{1}{T}\sum_{t=0}^{T-1}\left(V\left(W_n(t)+e_nH_n(t)+g_n(t)\right)\right.\nonumber\\
      &\left.+Q_n(t_f^{(n)})\left(\overline{\mu}^*_n-\mu_nH_n(t)\right)
       +(t-t_f^{(n)})B_0\right)\nonumber\\
\leq&V\left(\overline{W}_n^*+\overline{E}_n^*+\overline{G}_n^*\right)+\Psi_n,\label{DPP_sub_bound_2}
\end{align}
where $B_0=\frac{1}{2}(R_{\max}+\mu_{\max})\mu_{\max}$ as defined below \eqref{server_decision}, the first inequality follows from the fact that for any $t\in\left(t_f^{(n)},~t_{f+1}^{(n)}\right)$,
\begin{align*}
&Q_n(t)\left(\overline{\mu}^*_n-\mu_nH_n(t)\right)\nonumber\\
\leq& Q_n(t_f^{(n)})\left(\overline{\mu}^*_n-\mu_nH_n(t)\right)\\
      &+(Q_n(t)-Q_n(t_f^{(n)}))\left(\overline{\mu}^*_n-\mu_nH_n(t)\right)\\
\leq& Q_n(t_f^{(n)})\left(\overline{\mu}^*_n-\mu_nH_n(t)\right)\\
      &+\sum_{t=t_f^{(n)}}^{t_{f+1}^{(n)}-1}(R_n(t)-\mu_nH_n(t))\left(\overline{\mu}^*_n-\mu_nH_n(t)\right)\\
\leq& Q_n(t_f^{(n)})\left(\overline{\mu}^*_n-\mu_nH_n(t)\right)+(t-t_f^{(n)})B_0,
\end{align*}
and the second inequality follows from lemma \ref{true_time_average}. Substitute \eqref{DPP_sub_bound_1} and \eqref{DPP_sub_bound_2} into \eqref{DPP_bound} gives
\begin{align}\label{penultimate_step}
\limsup_{T\rightarrow\infty}\frac{1}{T}\sum_{t=0}^{T-1}P(t)\leq& V\left(\overline{C}^*+\sum_{n=1}^N\left(\overline{W}_n^*+\overline{E}_n^*+\overline{G}_n^*\right)\right)\nonumber\\
&+B_3+\sum_{n=1}^N\Psi_n.
\end{align}
Finally, notice that by telescoping sums,
\begin{align*}
&\limsup_{T\rightarrow\infty}\frac{1}{T}\sum_{t=0}^{T-1}P(t)\\
=&\limsup_{T\rightarrow\infty}\left(\frac{V}{T}\sum_{t=0}^{T-1}\left(c(t)d(t)
+\sum_{n=1}^N\left(W_n(t)+e_nH_n(t)\right.\right.\right.\\
&\left.\left.\left.+g_n(t)\right)\right)+\frac{1}{2}\sum_{n=1}^NQ_n(T)^2\right)\\
\geq&V\cdot\limsup_{T\rightarrow\infty}\frac{1}{T}\sum_{t=0}^{T-1}\left(c(t)d(t)
+\sum_{n=1}^N\left(W_n(t)+e_nH_n(t)\right.\right.\\
&\left.\left.+g_n(t)\right)\right)
\end{align*}
Substitute above inequality into \eqref{penultimate_step} and divide $V$ from both sides give the desired result.
\end{proof}

\bibliographystyle{unsrt}
\bibliography{bibliography}

\end{document}